\documentclass[11pt]{article}

\date{}
\author{Gabriel Nivasch\footnote{\texttt{gabrieln@ariel.ac.il}. Department of Computer Science, Ariel University, Ariel, Israel.}}
\title{On the zone of a circle in an arrangement of lines\footnote{An extended abstract of this paper appeared in \emph{EuroComb 2015} (\emph{Electronic Notes in Discrete Mathematics} 49:221--231, 2015).}}

\usepackage{amsmath, amsthm, amsfonts}
\usepackage{graphicx}

\newcommand{\eps}{\varepsilon}
\newcommand{\R}{{\mathbb R}}

\DeclareMathOperator{\Ex}{Ex}

\newtheorem{theorem}{Theorem}[section]
\newtheorem{lemma}[theorem]{Lemma}
\newtheorem{observation}[theorem]{Observation}
\newtheorem{conjecture}[theorem]{Conjecture}
\newtheorem{corollary}[theorem]{Corollary}

\theoremstyle{definition}
\newtheorem{definition}[theorem]{Definition}

\theoremstyle{remark}
\newtheorem{remark}[theorem]{Remark}

\begin{document}
\maketitle

\begin{abstract}
Let $\mathcal L$ be a set of $n$ lines in the plane, and let $C$ be a convex curve in the plane, like a circle or a parabola. The \emph{zone} of $C$ in $\mathcal L$, denoted $\mathcal Z(C,\mathcal L)$, is defined as the set of all cells in the arrangement $\mathcal A(\mathcal L)$ that are intersected by $C$. Edelsbrunner et al.~(1992) showed that the complexity (total number of edges or vertices) of $\mathcal Z(C,\mathcal L)$ is at most $O(n\alpha(n))$, where $\alpha$ is the inverse Ackermann function. They did this by translating the sequence of edges of $\mathcal Z(C,\mathcal L)$ into a sequence $S$ that avoids the subsequence $ababa$. Whether the worst-case complexity of $\mathcal Z(C,\mathcal L)$ is only linear is a longstanding open problem.

Since the relaxation of the problem to pseudolines does have a $\Theta(n\alpha(n))$ bound, any proof of $O(n)$ for the case of straight lines must necessarily use geometric arguments.

In this paper we present some such geometric arguments. We show that, if $C$ is a circle, then certain configurations of straight-line segments with endpoints on $C$ are impossible. In particular, we show that there exists a Hart--Sharir sequence that cannot appear as a subsequence of $S$.

The Hart--Sharir sequences are essentially the only known way to construct $ababa$-free sequences of superlinear length. Hence, if it could be shown that every family of $ababa$-free sequences of superlinear-length eventually contains all Hart--Sharir sequences, it would follow that the complexity of $\mathcal Z(C,\mathcal L)$ is $O(n)$ whenever $C$ is a circle.
\end{abstract}

\section{Introduction}

Let $\mathcal L$ be a set of $n$ lines in the plane. The \emph{arrangement} of $\mathcal L$, denoted $\mathcal A(\mathcal L)$, is the partition of the plane into vertices, edges, and cells induced by $\mathcal L$. Let $C$ be another object in the plane. The \emph{zone} of $C$ in $\mathcal L$, denoted $\mathcal Z(C,\mathcal L)$, is defined as the set of all cells in $\mathcal A(\mathcal L)$ that are intersected by $C$. The \emph{complexity} of $\mathcal Z(C,\mathcal L)$ is defined as the total number of edges, or vertices, in it.

The celebrated \emph{zone theorem} states that, if $C$ is another line, then $\mathcal Z(C,\mathcal L)$ has complexity $O(n)$ (Chazelle et al.~\cite{CGL}; see also Edelsbrunner et al.~\cite{EGPPSS}, Matou\v sek~\cite{mat_DG}).

If $C$ is a convex curve, like a circle or a parabola, then $\mathcal Z(C,\mathcal L)$ is known to have complexity $O(n \alpha(n))$, where $\alpha$ is the very-slow-growing inverse Ackermann function (Edelsbrunner et al.~\cite{EGPPSS}; see also Bern et al.~\cite{BEPY}, Sharir and Agarwal~\cite{DS_book}). More specifically, the \emph{outer zone} of $\mathcal Z(C,\mathcal L)$ (the part that lies outside the convex hull of $C$) is known to have complexity $O(n)$, whereas the complexity of the \emph{inner zone} is only known to be $O(n \alpha(n))$. Whether the complexity of the inner zone is linear as well is a longstanding open problem~\cite{BEPY, DS_book}.

The gap between the upper and the lower bound is completely negligible for all practical purposes, but the question is interesting from a purely mathematical point of view.

In this paper we make progress towards proving that the inner zone of a circle in an arrangement of lines has linear complexity. Since we find it easier to work with a parabola than with a circle, throughout this paper we will take $C$ to be the parabola $y=x^2$. The two problems are equivalent by a projective transformation, as we will explain.

\subsection{Davenport--Schinzel sequences and their generalizations}

Let $S$ be a finite sequence of symbols, and let $s\ge 1$ be a parameter. Then $S$ is called a \emph{Davenport--Schinzel sequence of order $s$} if every two adjacent symbols in $S$ are distinct, and if $S$ does not contain any alternation $a\cdots b\cdots a\cdots b\cdots$ of length $s+2$ for two distinct symbols $a\neq b$. Hence, for $s=1$ the ``forbidden pattern" is $aba$, for $s=2$ it is $abab$, for $s=3$ it is $ababa$, and so on.

The maximum length of a Davenport--Schinzel sequence of order $s$ that contains only $n$ distinct symbols is denoted $\lambda_s(n)$. For $s\le 2$ we have $\lambda_1(n) = n$ and $\lambda_2(n) = 2n-1$. However, for fixed $s\ge 3$, $\lambda_s(n)$ is slightly superlinear in $n$.

\paragraph{DS sequences of order $3$}
The case $s=3$ is the one most relevant to us. Hart and Sharir~\cite{HS} (see also~\cite{yo_DS,DS_book}) constructed a family of sequences that achieve the lower bound\footnote{The bound claimed in~\cite{DS_book} is $\lambda_3(n) \ge (1/2) n\alpha(n)-O(n)$, because a factor of $2$ is lost in interpolation; this problem is fixed in~\cite{yo_DS}.} $\lambda_3(n) \ge n\alpha(n) - O(n)$; and they also proved the asymptotically matching upper bound $\lambda_3(n) \le O(n\alpha(n))$. Klazar~\cite{klazar} subsequently improved the upper bound to $\lambda_3(n) \le 2n\alpha(n) + O(n\sqrt{\alpha(n)})$ (recently, Pettie~\cite{pettie_sharp} improved the lower-order term to $O(n)$).

Nivasch~\cite{yo_DS} showed that $\lambda_3(n) \ge 2n\alpha(n) - O(n)$. Hence, $\lambda_3(n) = 2n\alpha(n) \pm O(n)$.  Nivasch's construction is an \emph{extension} of the Hart--Sharir construction, in the sense that Nivasch's sequences contain the Hart--Sharir sequences as subseqeunces.\footnote{This can be shown with an argument similar to that of Lemma~\ref{lem_structurally} below, which is beyond the scope of this paper.} Geneson~\cite{geneson} made a nice cosmetic improvement to Nivasch's construction.

\paragraph{DS sequences of higher orders}
For $s=4$ we have $\lambda_4(n) = \Theta(n\cdot 2^{\alpha(n)})$, and in general, $\lambda_s(n) = n\cdot 2^{\Theta(\mathrm{poly}(\alpha(n)))}$ for fixed $s\ge 4$, where the polynomial in the exponent is of degree roughly $s/2$. See Sharir and Agarwal~\cite{DS_book}, and subsequent improvements by Nivasch~\cite{yo_DS} and Pettie~\cite{pettie_sharp}.

\paragraph{Generalized DS sequences}
A \emph{generalized Davenport--Schinzel sequence} is one where the forbidden pattern is not restricted to be $abab\cdots$, but it can be any fixed subsequence $u$. In order for the problem to be nontrivial we must require $S$ to be \emph{$k$-sparse}---meaning, every $k$ adjacent symbols in $S$ must be pairwise distinct---where $k = \|u\|$ is the number of distinct symbols in $u$. For example, if we take $u=abcaccbc$, then $S$ must not contain any subsequence of the form $a\cdots b\cdots c\cdots a\cdots c\cdots c\cdots b\cdots c$ for $|\{a,b,c\}|=3$, and every three adjacent symbols in $S$ must be pairwise distinct.

We denote by $\Ex(u, n)$ the maximum length of a $k$-sparse, $u$-avoiding sequence $S$ on $n$ distinct symbols, where $k=\|u\|$. For every fixed forbidden pattern $u$, $\Ex(u,n)$ is at most slightly superlinear in $n$: $\Ex(u,n) =  O\bigl(n\cdot 2^{\mathrm{poly}(\alpha(n))}\bigr)$, where the polynomial in the exponent depends on $u$ (Klazar~\cite{klazar_genDS}, Nivasch~\cite{yo_DS}, Pettie~\cite{pettie_3}).

Similarly, if $U = \{u_1, u_2, \ldots, u_j\}$ is a set of patterns, then $\Ex(U, n)$ denotes the maximum length of a sequence that avoids all the patterns in $U$, is $k$-sparse for $k = \min\{\|u\| : u\in U\}$, and contains only $n$ distinct symbols.

\paragraph{Some relevant results on generalized DS sequences}
Let us mention some results on generalized DS sequences that are relevant to us:
\begin{itemize}
\item $\Ex(\{ababa, ab\,cac\,cbc\}, n) = \Theta(n\alpha(n))$ (Pettie~\cite{pettie_origins}). Indeed, the $ababa$-free sequences of Hart and Sharir~\cite{HS} avoid $ab\,cac\,cbc$ as well.\footnote{Spaces are just for clarity. The Hart--Sharir construction also avoids other patterns, such as $abcbdadbcd$ (Klazar~\cite{klazar_93}; see Pettie~\cite{pettie_origins}).} See Section~\ref{sec_ababa} below.

\item $\Ex(ab\,cacbc, n) = \Theta(n\alpha(n))$ (Pettie~\cite{pettie_matrix}). The lower bound is achieved by a modification of the Hart--Sharir construction, which does not avoid $ababa$ anymore.

\item It is unknown whether $\Ex(\{ababa, ab\,cacbc\},n)$ or $\Ex(\{ababa, ab\,cac\,cbc, (ab\,cac\,cbc)^R\}, n)$ are superlinear in $n$ (where $u^R$ denotes the reversal of $u$). We conjecture that they are both $O(n)$.
\end{itemize}

\paragraph{Applications of generalized DS sequences}
Generalized Davenport--Schinzel sequences have found a few applications. Cibulka and Kyn\v cl~\cite{CK} used them to bound the size of sets of permutations with bounded VC-dimension. Valtr~\cite{valtr}, Fox et al.~\cite{FPS}, and Suk and Walczak~\cite{Suk-W} have used Generalized DS sequences to bound the number of edges in graphs with no $k$ pairwise crossing edges; the papers \cite{valtr,FPS} use the ``$N$-shaped" forbidden pattern $a_1\cdots a_\ell\cdots a_1\cdots a_\ell$, and the papers \cite{FPS,Suk-W} use the forbidden pattern $(a_1\cdots a_\ell)^m$.

Pettie considered $\Ex(\{abababa, abaabba\}, n)$ for analyzing the deque conjecture for splay trees~\cite{pettie_splay}, and $\Ex(\{ababab, abbaabba\}, n)$ for analyzing the union of fat triangles in the plane~\cite{pettie_forbid}.

\subsection{Transcribing the zone into a Davenport--Schinzel sequence}\label{subsec_Lprime}

Let $\mathcal L$ be a set of $n$ lines in the plane, and let $C$ be a convex curve in the plane. We can assume without loss of generality that $C$ is either closed (like a circle) or unbounded in both directions (like a parabola), by prolonging $C$ if necessary. Thus, $C$ divides the plane into two regions, one of which equals the convex hull of $C$.

Here we recall the argument of Edelsbrunner et al.~\cite{EGPPSS} showing that the complexity of the part of $\mathcal Z(C, \mathcal L)$ that lies inside the convex hull of $C$ is $O(n \alpha(n))$.

If $C$ is unbounded in both directions then assume without loss of generality that it is $x$-monotone and it is the graph of a convex function, by rotating the whole picture if necessary.

Also assume general position for simplicity: No line of $\mathcal L$ is vertical, no two lines are parallel, no three lines are concurrent, no line is tangent to $C$, and no two lines intersect $C$ at the same point. (Perturbing $\mathcal L$ into general position can only increase the complexity of $\mathcal Z(C, \mathcal L)$.) We can also assume that every line of $\mathcal L$ intersects $C$, since otherwise the line would not contribute to the complexity of the inner zone of $C$.

Let $\mathcal L'$ be the set of $n$ segments obtained by intersecting each line of $\mathcal L$ with the convex hull of $C$. (If $C$ is unbounded then some elements of $\mathcal L'$ may actually be rays.)

Let $G$ be the \emph{intersection graph} of $\mathcal L'$, i.e.~the graph having $\mathcal L'$ as vertex set, and having an edge connecting two elements of $\mathcal L'$ if and only if they intersect. We can assume without loss of generality that $G$ is connected: If $G$ has several connected components, then we can separately bound the complexity produced by each one and add them up; this works because our desired bound is at least linear in $n$. Since $G$ is connected, all the bounded cells of the inner zone are simple (i.e.~they touch $C$ in a single interval); and if $C$ is unbounded then there are at most two upward-unbounded cells (bounded by the two infinite extremes of $C$).

If $C$ is closed then let $c_0$ be the topmost point of $C$; we will ignore the cell that contains $c_0$, since it has at most linear complexity (as any single cell does). If $C$ is unbounded then we will similarly ignore the up-to-two unbounded cells.

To bound the complexity of the remaining cells, we will traverse their boundary and transcribe it into a sequence in a certain way.

Every segment of $\mathcal L'$ has two sides, one of which will be called \emph{positive} and the other one \emph{negative}, as follows: If $C$ is closed, then the positive side is the one facing the point $c_0$ and the negative side is the other one; if $C$ is unbounded, then the positive side is the upper one and the negative side is the lower one.

If $C$ is closed, then let $c_1$ be the first endpoint of $\mathcal L'$ counterclockwise from $c_0$ along $C$, and let $c_2$ be the last endpoint. If $C$ is unbounded, then $c_1$ is defined as the leftmost endpoint of $\mathcal L'$, and $c_2$ as the rightmost endpoint.

We traverse the boundary of the inner zone of $C$ by starting at $c_1$, and walking around the boundary of the cells, as if the segments were walls which we touch with the left hand at all times, until we reach $c_2$. See Figure~\ref{fig_zone_tour}. We transcribe this tour into a sequence containing $3n$ distinct symbols as follows:

Each segment $a\in \mathcal L'$ is partitioned by the other segments into smaller pieces. We take two directed copies of each such piece. We call each such copy a \emph{sub-segment}. The sub-segments are directed counterclockwise around $a$; i.e. those above $a$ are directed leftwards, and those below $a$ are directed rightwards. Hence, our tour visits some of these sub-segments, in the directions we have given them, in a certain order.

For each segment $a$, the sub-segments of $a$ that are visited, are visited in counterclockwise order around $a$. We first visit some sub-segments on the positive side of $a$, then we visit some sub-segments on the negative side of $a$, and then we again visit some sub-segments on the positive side of $a$.

\begin{figure}
\centerline{\includegraphics{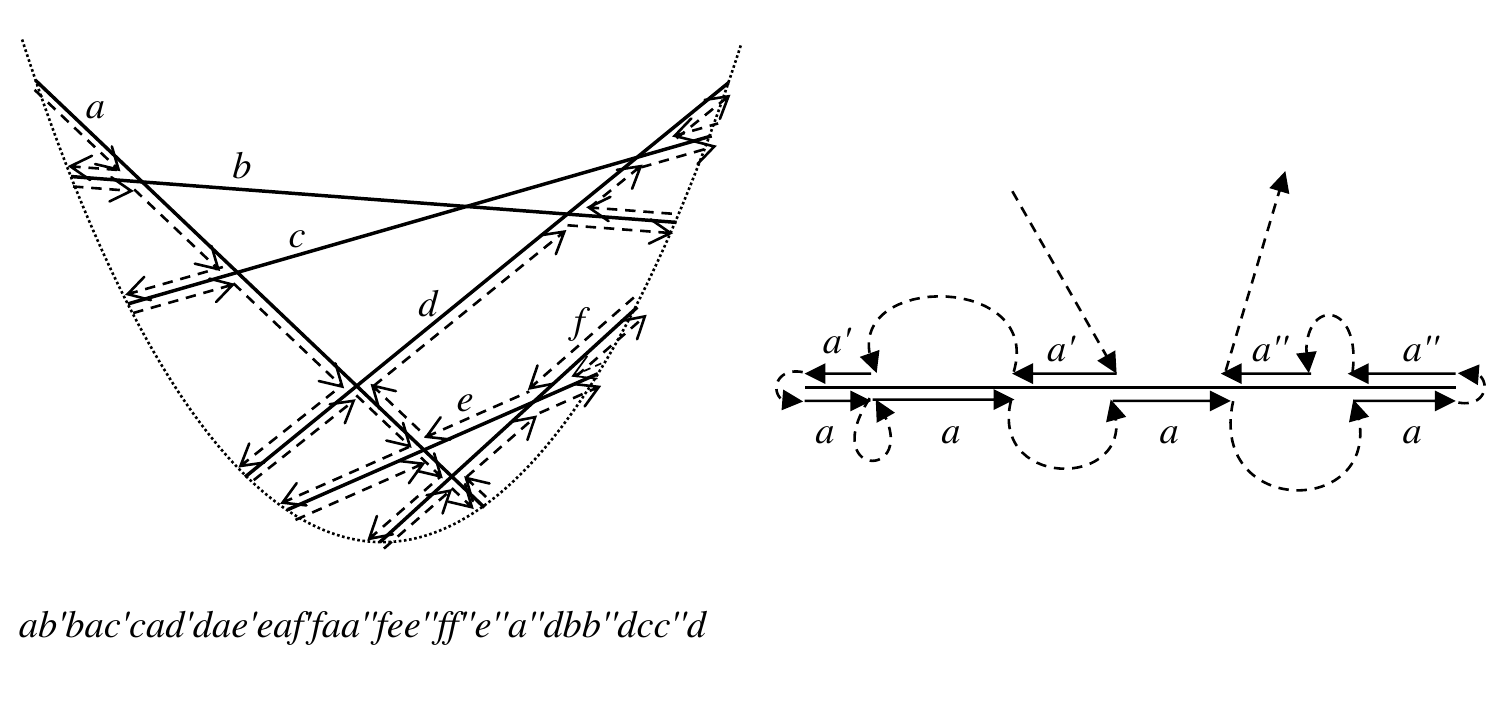}}
\caption{\label{fig_zone_tour}Traversing the boundary of the inner zone of $C$.}
\end{figure}

Sub-segments of the first type are transcribed as $a'$; sub-segments of the second type are transcribed as $a$, and sub-segments of the third type are transcribed as $a''$. See again Figure~\ref{fig_zone_tour}. Let $S'$ be the sequence resulting from the tour.

For each segment $a$, label its endpoints $L_a$ and $R_a$, such that $L_a$ is visited before $R_a$.\footnote{If $C$ is unbounded then $L_a$ is always the left endpoint of $a$.}

\begin{figure}
\centerline{\includegraphics{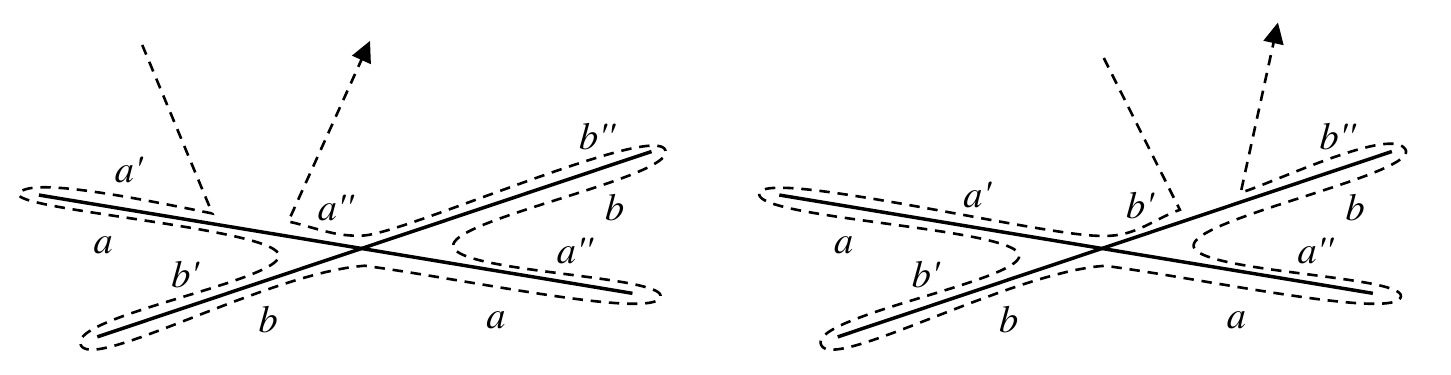}}
\caption{\label{fig_tour_2segments}Symbol alternations produced by two intersecting segments.}
\end{figure}

Let $a$, $b$ be two intersecting segments, such that $L_a$ is visited before $L_b$. Then the restriction of $S'$ to $\{a',a,a'',b',b,b''\}$ is of the form
\begin{equation*}
(a')^*\, a^*\, (b')^*\, b^*\, a^*\, (a'')^*\, b^*\, (b'')^*\, (a'')^* \quad \text{or}\quad (b')^*\, (a')^*\, a^*\, (b')^*\, b^*\, a^*\, (a'')^*\, b^*\, (b'')^*,
\end{equation*}
where ${}^*$ denotes zero or more repetitions. See Figure~\ref{fig_tour_2segments}.

Hence, the restriction of $S'$ to first-type symbols contains no alternation $abab$, and it contains no adjacent repetitions either, as can be easily seen. Hence, it is an order-$2$ DS-sequence and so it has linear length. The same is true for the restriction of $S'$ to third-type symbols.

Thus, the important part of the sequence $S'$ is its restriction to second-type symbols---those corresponding to the negative side of the segments. From now on we denote this subsequence $S$, and we call it the \emph{lower inner-zone sequence of $\mathcal Z(C, \mathcal L)$}.\footnote{Slight abuse of terminology. We will mainly deal with the case where $C$ is unbounded; in this case the negative side of a segment is always its lower side.} The sequence $S$ contains no alternation $ababa$, and it contains no adjacent repetitions, as can be easily seen. Hence, $S$ is an order-$3$ DS-sequence, and hence its length is at most $O(n\alpha(n))$.

\subsection{Relation to lower envelopes}

Lower envelopes are the original motivation for Davenport--Schinzel sequences. If $\mathcal F = \{f_1, \ldots,\allowbreak f_n\}$ is a collection of $n$ $x$-monotone curves in the plane (continuous functions $\R\to\R$), then the \emph{lower envelope} of $\mathcal F$ is their pointwise minimum (or the part that can be seen from the point $(0,-\infty)$), and the \emph{lower-envelope sequence} is the sequence of functions that appear in the lower envelope, from left to right. If the $f_i$'s are partially defined functions (say, each one is defined only on an interval of $\R$), then the definition is the same, except that the symbol ``$\infty$" might also appear in the lower-envelope sequence.

In our case, if $C$ is $x$-monotone, then the lower-envelope sequence of the set of segments $\mathcal L'$ is a subsequence of $S$: It contains only those parts that can be seen from $-\infty$. We shall denote this sequence by $N = N(\mathcal L')$.

The Hart--Sharir sequences can be realized as lower-envelope sequences of segments in the plane (Wiernik and Sharir~\cite{WS}; see also~\cite{mat_DG, DS_book}). However, it is unknown whether this is still possible if all the endpoints are required to lie on a circle/parabola (like our set $\mathcal L'$), or more generally on a convex curve. Sharir and Agarwal raise this question in~\cite[p.~112]{DS_book}. Proving a linear upper bound for the length of $N$ might be easier than for the length of $S$.

It is also not known whether the longer sequences of Nivasch~\cite{yo_DS} can be realized as lower-envelope sequences of segments. It is not even known whether there \emph{exists} an order-$3$ DS sequence that cannot be realized as a lower-envelope sequence of segments.

\subsection{From circles to conic sections}\label{subsec_conic}

Let us return to the zone problem. The special cases in which $C$ is a circle, a parabola, or a hyperbola, are all equivalent, as can be shown by suitable projective transformations: Let $\pi_1\subset \R^3$ be a plane that contains a set of lines $\mathcal L$ and a circle $C$. Let $K$ be a cone in $\R^3$ that intersects $\pi_1$ at $C$, and let $\pi_2$ be a plane that intersects $K$ at a parabola $C'$. Then, the projection through the apex of $K$ maps $\pi_1$ (with the exception of one line within $\pi_1$) into $\pi_2$, mapping lines into lines, and mapping $C$ (except for one point $p\in C$) into $C'$. We just have take care to choose $\pi_2$ so that no line of $\mathcal L$ passes through $p$.

More concretely, the projective transformation $(x,y)\mapsto\bigl(\frac{x}{1-y},\frac{1+y}{1-y}\bigr)$ maps the unit circle $x^2+y^2=1$ (except for the point $p=(0,1)$) into the parabola $y=x^2$, mapping lines into lines.

The case of a hyperbola is handled similarly. First, note that all hyperbolas are equivalent under affine transformations. Hence, choose $\pi_2$ so that it intersects the cone $K$ at a hyperbola $C'$, such that almost all of $C$ is mapped to one branch of $C'$, and only a tiny portion of $C$, which does not intersect any line of $\mathcal L$, is mapped to the other branch of $C'$.

Even though the most natural formulation of the problem involves a circle, in this paper we will work with a parabola, since we find it easier to work with.

\subsection{The case of pseudolines and the need for geometric arguments}

If we relax the problem and allow $\mathcal L$ to consist of $x$-monotone \emph{pseudolines} ($x$-monotone curves that pairwise intersect at most once and intersect $C$ at most twice), then $\mathcal Z(C,\mathcal L)$ can have complexity $\Theta(n\alpha(n))$. Indeed, in this setting \emph{every} order-$3$ DS-sequence can appear as a subsequence of a lower-envelope sequence $N(\mathcal L')$. To see this, first note that every order-$3$ DS-sequence can appear as a subsequence of a lower-envelope sequence of $x$-monotone pseudosegments~\cite{DS_book}; see Figure~\ref{fig_pseudolines} for an example. Furthermore, in this construction, all segment endpoints are visible from $-\infty$. Hence, we can enclose the construction in a circle $C$, and prolong each pseudosegment $\ell$ into an $x$-monotone pseudoline by adding two very steeply decreasing rays on the two sides of $\ell$.

\begin{figure}
\centerline{\includegraphics{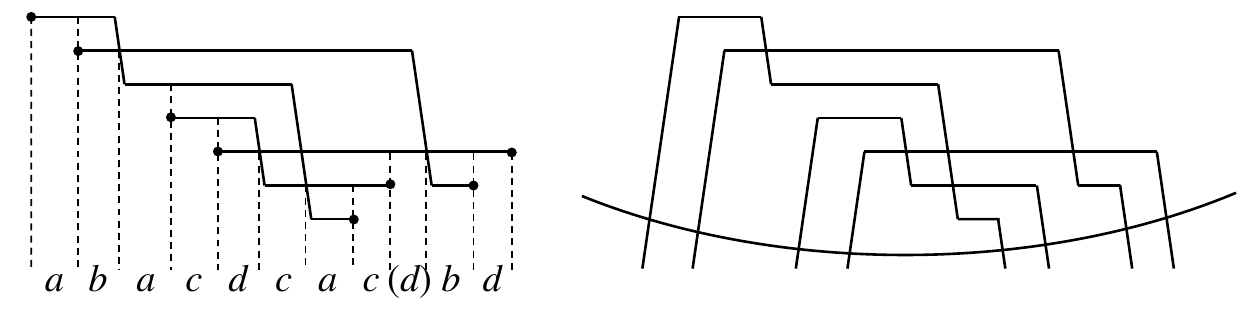}}
\caption{\label{fig_pseudolines}Left: Realizing the sequence $u=abacdcacbd$ as a lower-envelope sequence of pseudosegments. Note that in this case the technique produces a supersequence of $u$. Right: Adding a convex curve and prolonging the pseudosegments into pseudolines.}
\end{figure}

Therefore, if, as we conjecture, the bound for the case of straight lines is only $O(n)$, then any proof must necessarily use geometric arguments, and not merely combinatorial ones.

\subsection{Our results}\label{subsec_our_results}

In this paper we offer some evidence for the following conjecture, and make some progress towards proving it:

\begin{conjecture}\label{conj}
If $\mathcal L$ is a set of $n$ lines and $C$ is a circle, then the lower inner-zone sequence $S$ of $\mathcal Z(C, \mathcal L)$ has length $O(n)$, and hence $\mathcal Z(C, \mathcal L)$ has at most linear complexity.
\end{conjecture}

Our technique consists of first finding segment configurations that are geometrically impossible, and then finding $ababa$-free sequences that \emph{force} these configurations. We say that an $ababa$-free sequence $u$ \emph{forces} a segment configuration $T$, if, for every family of segments $\mathcal L'$ (as in Section~\ref{subsec_Lprime}) whose lower inner-zone sequence contains $u$ as a subsequence, $\mathcal L'$ contains a subfamily combinatorially equivalent to $T$.

Thus, we first show in Section~\ref{sec_useless_forbidden} that a certain, relatively simple configuration of eleven segments is impossible. Then we show that this configuration is forced by a pattern $u$ of length $33$. It follows that the lower inner-zone sequence $S$ avoids $u$. This result, however, is useless for establishing Conjecture~\ref{conj}, since $u$ contains both $ab\,cac\,cbc$ and its reversal. Therefore, by the above-mentioned result of Pettie, the Hart--Sharir construction avoids both $u$ and $u^R$ (which is actually the same as $u$), and so $\Ex(\{ababa, u, u^R\}, n) = \Theta(n\alpha(n))$.

Section~\ref{sec_useless_forbidden} is just a warmup for Sections~\ref{sec_forbidden_2} and~\ref{sec_ababa}. In Section~\ref{sec_forbidden_2} we construct another impossible segment configuration $X$, this time with $173$ segments. We could construct a pattern $u'$ that, as a consequence, cannot occur in $S$, but we abstain from doing so. Instead, we show directly in Section~\ref{sec_ababa} that the Hart--Sharir sequences eventually force the configuration $X$.

Our results in Sections~\ref{sec_forbidden_2} and~\ref{sec_ababa} were obtained as follows: Matou\v sek~\cite{mat_DG} and Sharir and Agarwal~\cite{DS_book} describe a construction by P. Shor of segments in the plane whose lower-envelope sequences are the Hart--Sharir sequences. We tried to force the segment endpoints in the construction to lie on the parabola $C$, and we reached a contradiction. The impossible configuration $X$ of Section~\ref{sec_forbidden_2} is the best way we found to isolate the contradiction. We elaborate more on this at the end of Section~\ref{sec_ababa}.

Next, in Section~\ref{sec_conj_only_DS} we present some directions for further work on the problem: We formulate a conjecture regarding generalized DS sequences which, if true, would imply Conjecture~\ref{conj}. We also explain how our conjecture relates to previous research on generalized DS sequences.

Finally, in Section~\ref{sec_discussion} we conclude by listing some related open problems.

\section{Preliminaries}\label{sec_preliminaries}

Throughout this and the following sections $C$ will denote the parabola $y=x^2$, $\mathcal L'$ will denote a set of $n$ segments with endpoints on $C$, and $S$ will denote the corresponding lower inner-zone sequence. As we said, we assume that no two segments have the same endpoint, and that the intersection graph of $\mathcal L'$ is connected.

Recall that the left and right endpoints of a segment $a\in\mathcal L'$ are denoted $L_a$ and $R_a$, respectively. Whenever we say that a sequence of endpoints appear in a certain order, we mean from left to right.

Let $u$ be an $ababa$-free sequence in which, for simplicity, each symbol appears at least twice. Then, we define its \emph{endpoint sequence} $E(u)$ by replacing, for each symbol $a$ in $u$, its first occurrence by $L_a$ and its last occurrence by $R_a$, and deleting all other occurrences of $a$. For example,
\begin{equation*}
E(abacadcdbd) = L_a\, L_b\, L_c\, R_a\, L_d\, R_c\, R_b\, R_d.
\end{equation*}

It is clear that, if $S$ is the lower inner-zone sequence of $\mathcal L'$, then the order of the endpoints of $\mathcal L'$ is exactly $E(S)$. However, if $u$ is a subsequence of $S$, then $E(u)$ is not necessarily a subsequence of $E(S)$; meaning, $E(S)$ does not necessarily respect the order of the symbols in $E(u)$. For example, if $u=abba$, so the order of the left endpoints in $E(u)$ is $L_a, L_b$, it is still possible for their order in $E(S)$ to be $L_b, L_a$.

Nevertheless, we now give a sufficient condition for guaranteeing that $E(u)$ is a subsequence of $E(S)$:

\begin{definition}\label{def_clamped}
A symbol $a$ in a sequence $u$ is said to be \emph{left-clamped} if $u$ contains the subsequence $baba$, where $b$ is the symbol immediately preceding the first $a$ in $u$. Similarly, the symbol $a$ is \emph{right-clamped} if $u$ contains the subsequence $abab$, where $b$ is the symbol immediately following the last $a$ in $u$.
\end{definition}

\begin{lemma}\label{lemma_clamped}
If all the symbols in $u$, except for the very first symbol, are left-clamped, and all the symbols except for the very last symbol are right-clamped, and if $S$ is an $ababa$-free supersequence of $u$, then $E(u)$ is a subsequence of $E(S)$.
\end{lemma}

\begin{proof}
Consider an adjacent pair of symbols $Q_a$, $Q'_b$ in $E(u)$, where each of $Q$, $Q'$ is either $L$ or $R$. We claim that their order in $E(S)$ is also $Q_a$, $Q'_b$.

If $Q$ is $L$ and $Q'$ is $R$, then trivially $L_a$ also precedes $R_b$ in $E(S)$.

If $Q$ is $R$, then let $c$ be the symbol immediately following the last $a$ in $u$ (note that $c$ is not necessarily $b$). Since $a$ is right-clamped, $S$ cannot contain any $a$ after the occurrence of $b$ that gives rise to $Q'_b$; otherwise, $S$ would contain the forbidden pattern $acaca$.

Similarly, if $Q'$ is $L$, let $d$ be the symbol immediately preceding the first $b$ in $u$. Since $b$ is left-clamped, $S$ cannot contain any $b$ before the occurrence of $a$ that gives rise to $Q_a$; otherwise, $S$ would contain $bdbdb$.

Hence, whether $(Q,Q')$ equals $(R,L)$, $(R,R)$, or $(L,L)$, there is no way for $Q_a$, $Q'_b$ to change order in $E(S)$.
\end{proof}

Two segments $a, b$ intersect if and only if their endpoints appear in the order $L_a\, L_b\, R_a\, R_b\,$ or $L_b\, L_a\, R_b\, R_a$.

If $a_1, \ldots, a_m$ are segments whose endpoints appear in the order $L_{a_1}\, \cdots\, L_{a_m}\,\allowbreak R_{a_1}\, \cdots\, R_{a_m}$, then they pairwise intersect. If the intersection points $a_m \cap a_{m-1}, \ldots, a_3\cap a_2, a_2 \cap a_1$ appear in this order from left to right, then we say that the segments \emph{intersect concavely}. If the intersection points appear in the reverse order, then we say that the segments \emph{intersect convexly}. See Figure~\ref{fig_concavely}.

If the segments $a_1, a_2, \ldots, a_n$ intersect concavely (or convexly), and $1\le i_1 < i_2 < \cdots < i_k \le n$ are increasing indices, then $a_{i_1}, a_{i_2}, \ldots, a_{i_k}$ also intersect concavely (or convexly).

\begin{figure}
\centerline{\includegraphics{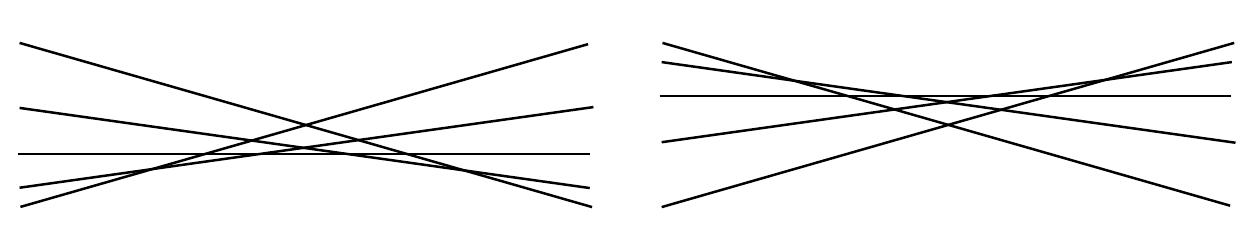}}
\caption{\label{fig_concavely}Left: Segments intersecting concavely. Right: Segments intersecting convexly.}
\end{figure}

\begin{observation}\label{obs_Nshaped}
If $S$ contains the ``$N$-shaped" subsequence $12\cdots m\cdots 212\cdots m$, then the corresponding segments must have endpoints in the order $L_1\,\cdots\, L_m\,\allowbreak R_1\,\cdots\, R_m$, and must intersect concavely.
\end{observation}

\begin{proof}[Proof sketch]
The general case follows from the case $m=3$.
\end{proof}

In Sections~\ref{sec_useless_forbidden} and~\ref{sec_forbidden_2}, we will specify some \emph{segment configurations} by listing the order of their endpoints, and by specifying that some subsets of segments must intersect concavely. We will prove that some segment configurations are geometrically impossible.

\subsection{Geometric properties of the parabola}

We now present some simple geometric properties of the parabola $C$ and straight-line segments. These properties lie at the heart of the distinction between the cases of straight lines and pseudolines, as we explained in the Introduction.

\begin{observation}
Let $a, b\in \R$ be fixed. Then the affine transformation $m:\R^2\to\R^2$ given by $m(x, y) = (ax+b,2abx+a^2y+b^2)$ maps the parabola $C$ to itself and keeps vertical lines vertical. Therefore, we are free to horizontally translate and scale the set of $x$-coordinates of the segment endpoints $\mathcal L'$, without affecting the resulting lower inner-zone sequence $S$ or the lower-envelope sequence $N$.\footnote{This observation would be useful in \emph{lower-bound constructions}, and hence, it is not used in the paper. We include it here just for the sake of completeness.}
\end{observation}

\begin{lemma}\label{lem_parabola}
Let $a$, $b$, $c$, $d$ be four points on the parabola $C$, having increasing $x$-coordinates $a_x < b_x < c_x < d_x$. Let $z = ac \cap bd$. Define the horizontal distances $p = b_x - a_x$, $q = d_x - c_x$, $r = z_x - b_x$, $s = c_x - z_x$. Then $p/q=r/s$. See Figure~\ref{fig_lem_parabola}, left.
\end{lemma}

\begin{figure}
\centerline{\includegraphics{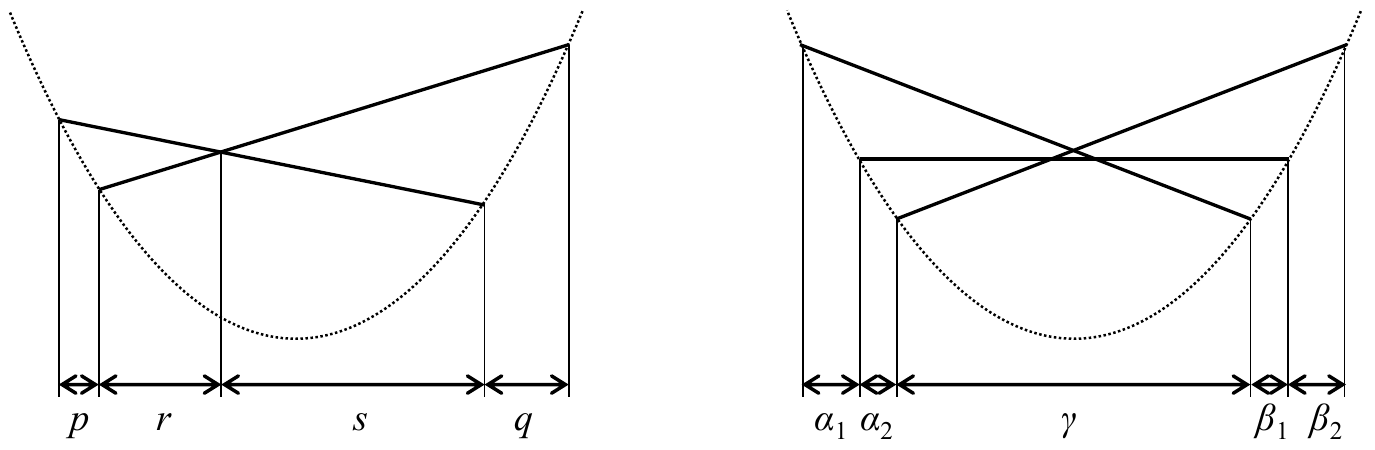}}
\caption{\label{fig_lem_parabola}Left: We have $p/q = r/s$. Right: Three segments intersecting concavely.}
\end{figure}

\begin{proof}
This can be shown directly by a slightly cumbersome algebraic calculation.

An alternative, more insightful proof uses elementary geometry and a limiting argument:

Let $C$ be not a parabola but a unit circle. Let $\alpha$ be a very small angle, and let $C_1$ be the arc of $C$ measuring angle $\alpha$ that is centered around the lowest point of $C$. Let $a$, $b$, $c$, $d$ be four points on $C_1$, in this order from left to right, and let $z=ac\cap bd$. Then, by the intersecting chords theorem, we have $ab/cd = bz/cz$. Since all the considered segments are almost horizontal, their length is almost equal to their $x$-projection. If we affinely stretch $C_1$ horizontally and vertically so its bounding box has width $1$ and height $1$, then it will almost match a parabola: At the limit as $\alpha\to 0$, the stretched arc pointwise converges to a parabolic segment. This affine transformation preserves the ratios between the horizontal projections, so the result follows.
\end{proof}

Lemma~\ref{lem_parabola} is actually part of a more general correspondence between circles and parabolas; see Yaglom \cite{yaglom}.

\begin{lemma}\label{lem_ratios}
Let $a,b,c,d,e,f$ be six points on the parabola $C$, listed by increasing $x$-coordinate. Suppose the segments $ad$, $be$, $cf$ intersect concavely. Define $\alpha_1 = b_x - a_x$, $\alpha_2 = c_x-b_x$, $\gamma = d_x-c_x$, $\beta_1 = e_x-d_x$, $\beta_2 = f_x-e_x$. Then:
\begin{enumerate}
\item $\alpha_1/\beta_1 > \alpha_2/\gamma$ and $\beta_2/\alpha_2 > \beta_1/\gamma$;
\item $\alpha_1/\beta_1 > \alpha_2/\beta_2$;
\item $\beta_1 < \beta_2$ or $\alpha_2 < \gamma+\beta_1+\beta_2$ (or both).
\end{enumerate}
See Figure~\ref{fig_lem_parabola}, right.
\end{lemma}

\begin{proof}
Let $g = be \cap cf$, $h = ad \cap be$. Subdivide $\gamma$ into $\gamma_1 = g_x - c_x$, $\gamma_2 = h_x - g_x$, $\gamma_3 = d_x - h_x$. By Lemma~\ref{lem_parabola} we have
\begin{equation*}
\frac{\alpha_1}{\beta_1} = \frac{\alpha_2 + \gamma_1+\gamma_2}{\gamma_3},\qquad
\frac{\alpha_2}{\beta_2} = \frac{\gamma_1}{\gamma_2 + \gamma_3 + \beta_1};
\end{equation*}
from which the first two claims follow.

By the first claim we have
\begin{equation*}
\frac{\beta_1}{\beta_2} < \frac{\gamma}{\alpha_2} < \frac{\gamma+\beta_1 + \beta_2}{\alpha_2};
\end{equation*}
hence, if $\beta_1/\beta_2$ is larger than $1$, then so is $(\gamma+\beta_1+\beta_2)/\alpha_2$, implying the third claim.
\end{proof}

\begin{figure}
\centerline{\includegraphics{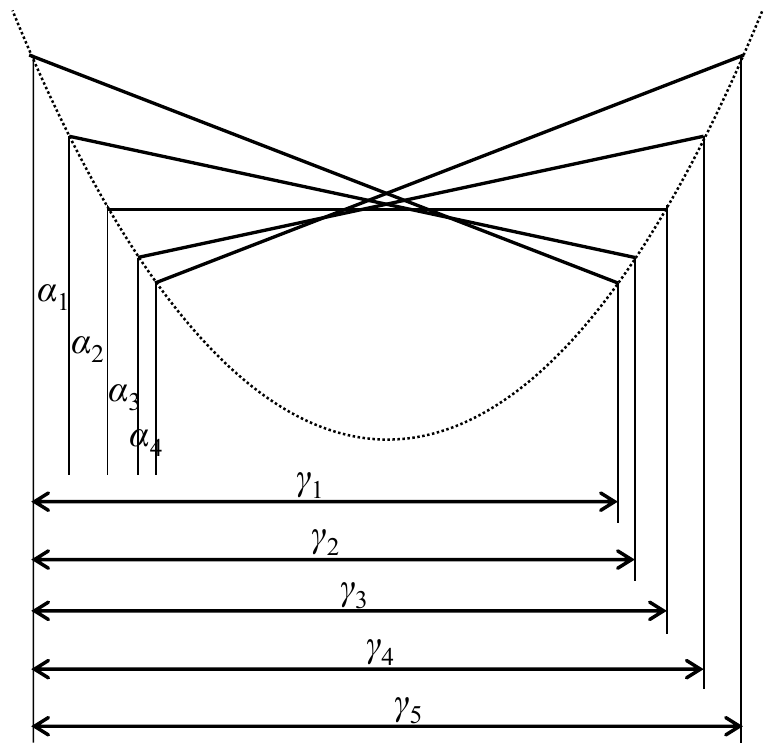}}
\caption{\label{fig_wide_fan}Illustration for Lemma~\ref{lemma_wide_fan} (picture not to scale).}
\end{figure}

\begin{definition}
Let $s_1, s_2, \ldots, s_m$ be segments whose endpoints appear in the order $L_{s_1}\, \cdots\, L_{s_m}\,\allowbreak R_{s_1}\, \cdots\, R_{s_m}$. These segments are called a \emph{wide set} if their $x$-coordinates satisfy $R_{s_kx} - L_{s_1x} > 2(R_{s_{k-1}x} - L_{s_1x})$ for each $2\le k\le m$.
\end{definition}

\begin{lemma}\label{lemma_wide_fan}
Let $s_1, \ldots, s_m$ be a wide set of segments that intersect concavely, and let $\alpha_k = L_{s_{k+1}x} - L_{s_kx}$ for $1\le k \le m-1$. Then $\alpha_k > \alpha_{k+1} + \cdots + \alpha_{m-1}$ for each $1\le k \le m-2$.
\end{lemma}

\begin{proof}
Let $\gamma_k = R_{s_kx} - L_{s_1x}$ for $1\le k\le m$. See Figure~\ref{fig_wide_fan}. We are given that $\gamma_k > 2\gamma_{k-1}$ for each $2\le k\le m$. Applying the first claim of Lemma~\ref{lem_ratios} to segments $s_k$, $s_{k+1}$, $s_m$, we get
\begin{equation*}
\frac{\alpha_k}{\alpha_{k+1} + \cdots + \alpha_{m-1}} > \frac{\gamma_{k+1}-\gamma_k}{\gamma_k - \sum \alpha_i} > \frac{\gamma_{k+1} - \gamma_k}{\gamma_k} > 1.
\end{equation*}
The claim follows.
\end{proof}

\section{Warmup: A simple but useless impossible configuration}\label{sec_useless_forbidden}

\begin{theorem}\label{thm_useless_impossible}
Let $a,b,c,d,e,1,2,3,4,8,9$ be eleven segments with endpoints on the parabola $C$, in left-to-right order
\begin{equation}\label{eq_useless_LR}
L_8\, L_1\, L_a\, L_b\, L_2\, R_8\, L_c\, L_d\, R_1\, R_2\, L_e\, R_a\, L_3\, L_4\, R_b\, R_c\, L_9\, R_3\, R_d\, R_e\, R_4\, R_9.
\end{equation}
Then, it is impossible for segments $8,1,2$ to intersect concavely, segments $3,4,9$ to intersect concavely, and segments $a,b,c,d,e$ to intersect concavely, all at the same time.
\end{theorem}

\begin{proof}
Suppose for a contradiction that $a,\ldots,9$ are segments satisfying all these properties. The intersection point of segments $1$ and $2$, which we shall call $A$, must lie left of $R_8$, and the intersection point of segments $3$ and $4$, which we shall call $B$, must lie right of $L_9$. See Figure~\ref{fig_2parabolas}.

\begin{figure}
\centerline{\includegraphics{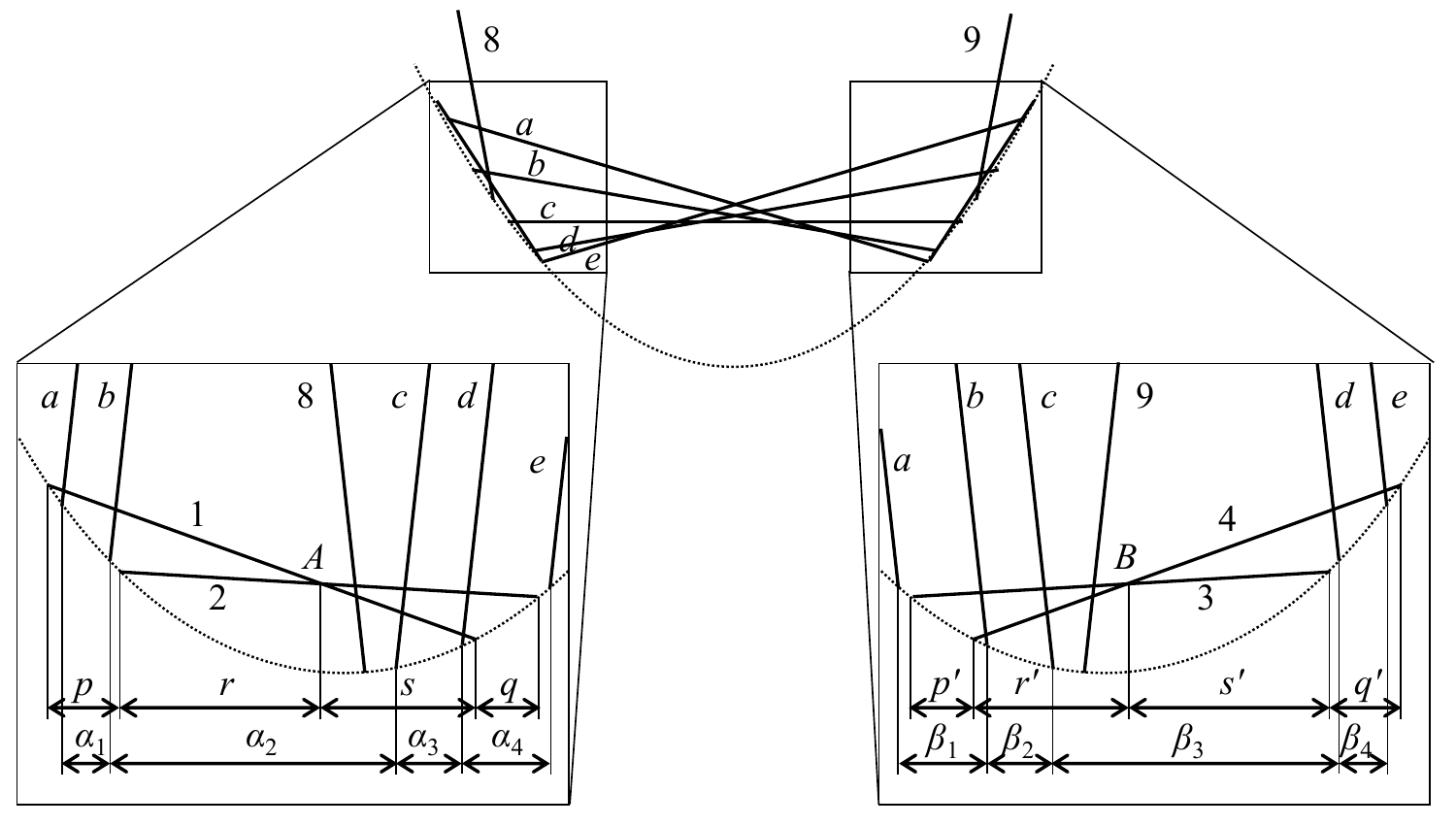}}
\caption{\label{fig_2parabolas}An impossible configuration of segments.}
\end{figure}

Define:
\begin{align*}
\alpha_1 &=L_{bx} - L_{ax}, & \beta_1 &=R_{bx}-R_{ax},\\
\alpha_2 &=L_{cx} - L_{bx}, & \beta_2 &=R_{cx}-R_{bx},\\
\alpha_3 &=L_{dx} - L_{cx}, & \beta_3 &=R_{dx}-R_{cx},\\
\alpha_4 &=L_{ex} - L_{dx}, & \beta_4 &=R_{ex}-R_{dx}.
\end{align*}
Segments $a,b,c,d,e$ must intersect concavely, so by the second claim of Lemma~\ref{lem_ratios}, we must have
\begin{equation}\label{eq_dec_ratios}
\frac{\alpha_1}{\beta_1} > \frac{\alpha_2}{\beta_2} > \frac{\alpha_3}{\beta_3} > \frac{\alpha_4}{\beta_4}.
\end{equation}
We will show, however, that this is impossible.
Define:
\begin{align*}
p &= L_{2x} - L_{1x}, & p' &= L_{4x} - L_{3x},\\
r &= A_x    - L_{2x}, & r' &= B_x    - L_{4x},\\
s &= R_{1x} - A_x   , & s' &= R_{3x} - B_x   ,\\
q &= R_{2x} - R_{1x}, & q' &= R_{4x} - R_{3x}.
\end{align*}
Then,
\begin{align*}
\alpha_1 &< p, & \beta_1 &> p',\\
\alpha_2 &> r, & \beta_2 &< r',\\
\alpha_3 &< s, & \beta_3 &> s',\\
\alpha_4 &> q, & \beta_4 &< q'.
\end{align*}
Furthermore, by Lemma~\ref{lem_parabola} we have $p/q = r/s$, $p'/q' = r'/s'$. Hence,
\begin{equation*}
\frac{\alpha_1\alpha_3}{\beta_1\beta_3} < \frac{ps}{p's'} = \frac{qr}{q'r'} < \frac{\alpha_2\alpha_4}{\beta_2\beta_4},
\end{equation*}
contradicting (\ref{eq_dec_ratios}).
\end{proof}

\begin{corollary}\label{cor_useless_forbidden}
Let $S$ be the lower inner-zone sequence of the parabola $C$ in an arrangement of lines. Then $S$ cannot contain a subsequence isomorphic to
\begin{equation*}
u = 81ab12181cd12dedcbab34bc49434de49.
\end{equation*}
\end{corollary}

\begin{proof}
We have
\begin{equation*}
E(u) = L_8\, L_1\, L_a\, L_b\, L_2\, R_8\, L_c\, L_d\, R_1\, R_2\, L_e\, R_a\, L_3\, L_4\, R_b\, R_c\, L_9\, R_3\, R_d\, R_e\, R_4\, R_9,
\end{equation*}
exactly matching (\ref{eq_useless_LR}). Furthermore, as a tedious examination shows, all symbols but $8$ are left-clamped in $u$, and all symbols but $9$ are right-clamped in $u$. Furthermore,
 $u$ contains the $N$-shaped subsequences $8121812$, $3494349$, and $abcdedcbabcde$. Hence, by Lemma~\ref{lemma_clamped} and Observation~\ref{obs_Nshaped}, $u$ forces the impossible segment configuration of Theorem~\ref{thm_useless_impossible}.
\end{proof}

(We obtained the sequence $u$ by simply taking the lower inner-zone sequence of the configuration of Theorem~\ref{thm_useless_impossible}, and removing from it unnecessary symbols.)

If we are only interested in a pattern avoided by $N$, the lower-envelope sequence, then we can omit the symbols $8$ and $9$ from $u$. Their only role is preventing the intersection points $A$ and $B$ from ``hiding" above the segment $c$.

Unfortunately, as we said in the Introduction, the forbidden pattern $u$ is useless for establishing Conjecture~\ref{conj}, since $u$ contains both $ab\, cac\, cbc$ and its reversal (e.g., $be\,4b4\,4e4$, $1a1\,1d1\,ad$). Furthermore, there does not seem to be a simple way to ``fix" $u$.

\section{A more promising impossible configuration}\label{sec_forbidden_2}

In this section we construct a $173$-segment impossible configuration $X$. Then, in Section~\ref{sec_ababa} we will show that the Hart--Sharir sequences eventually force the configuration $X$.

We will now work with endpoint sequences in which some contiguous subsequences (\emph{blocks}) that contain only left endpoints are designated as \emph{special blocks}. It will always be the case that all the special blocks in a sequence have the same length. We denote special blocks by enclosing them in parentheses.

We define an operation on endpoint sequences called \emph{endpoint shuffling}. This operation is derived from the \emph{shuffling} operation used to construct the Hart--Sharir sequences, which we will present in Section~\ref{sec_ababa} below.

Let $A$ be a sequence that has $k$ special blocks of length $m$, and let $B$ be a sequence that has $\ell$ special blocks of length $k$. Then the \emph{endpoint shuffle} of $A$ and $B$, denoted $A \circ B$, is a new sequence having $k\ell$ special blocks of length $m+1$, formed as follows: We make $\ell$ copies of $A$ (one for each special block of $B$), each one having ``fresh" symbols that do not occur in $B$ nor in any other copy of $A$.

For each special block $\Gamma_i = (L_1\, \cdots\, L_k)$ in $B$, $1\le i\le\ell$, let $A_i$ be the $i$th copy of $A$. We insert each $L_j$ at the end of the $j$th special block of $A_i$. Then we insert the resulting sequence in place of $\Gamma_i$ in $B$. The result of all these replacements is the desired sequence $A\circ B$.

For example, let 
\begin{equation*}
A = (L_a)\, (L_b)\, (L_c)\, R_a\, R_b\, R_c, \qquad B = (L_1\, L_2\, L_3)\, (L_4\, L_5\, L_6)\, R_1\, R_4\, R_2\, R_5\, R_3\, R_6.
\end{equation*}
Then,
\begin{multline*}
A\circ B = (L_a\, L_1)\, (L_b\, L_2)\, (L_c\, L_3)\, R_a\, R_b\, R_c\, (L_{a'}\, L_4)\\
(L_{b'}\, L_5)\, (L_{c'}\, L_6)\, R_{a'}\,R_{b'}\, R_{c'}\, R_1\, R_4\, R_2\, R_5\, R_3\, R_6.
\end{multline*}

Now, define the following endpoint sequences:
\begin{align*}
F_m &= (L_1\cdots L_m)\quad R_1\cdots R_m, \qquad m\ge 1;\\
Z_m &= L_a\, L_b\quad (L_1\cdots L_m)\quad R_1\cdots R_m\quad L_c\, R_a\\
&\qquad\qquad(L_{m+1}\cdots L_{2m})\quad R_{m+1}\cdots R_{2m}\quad R_b\, R_c, \qquad m\ge 1;\\
Y &= L_d\, L_e\, ()\, ()\, L_f\, R_d\, ()\, ()\, R_e\, R_f\, (),
\end{align*}
where $Y$ has five empty special blocks.

Recall from Section~\ref{sec_preliminaries} the definition of a \emph{wide} set of segments. We will now show how, by shuffling the $Z$ and $F$ sequences in the appropriate way, we can create, for every $n$, a configuration that contains a wide set of $n$ segments.\footnote{With a slight abuse of terminology, we do not always distinguish between endpoint sequences and the corresponding segment configurations.}

Consider the segment configuration
\begin{equation*}
T_n = (\cdots (((Z_1 \circ Z_2)\circ Z_4)\circ Z_8)\circ \cdots Z_{2^{n-1}})\circ F_{2^n}.
\end{equation*}
$T_n$ contains $2^{n-k-1}$ copies of $Z_{2^k}$ for each $0\le k<n$, plus one copy of $F_{2^n}$. Each copy of $Z_m$ contains its own triple of segments $a,b,c$. These segments resemble a cat's whiskers, so let us call each such triple $a,b,c$ a \emph{whisker} for short. The whiskers of $T_n$ are nested inside one another in the form of a complete binary tree of height $n$.

In addition, each copy $F_m$ and $Z_m$ contains segments labeled $1,\ldots, m$ and $m+1,\ldots,2m$, which we shall call \emph{numeric segments}. In total, $T_n$ contains $(n+1) 2^n$ numeric segments, whose left endpoints appear in $2^n$ special blocks of length $n+1$. Each such special block appears where we would expect to find the whiskers of level $n+1$ of the binary tree.

\begin{figure}
\centerline{\includegraphics{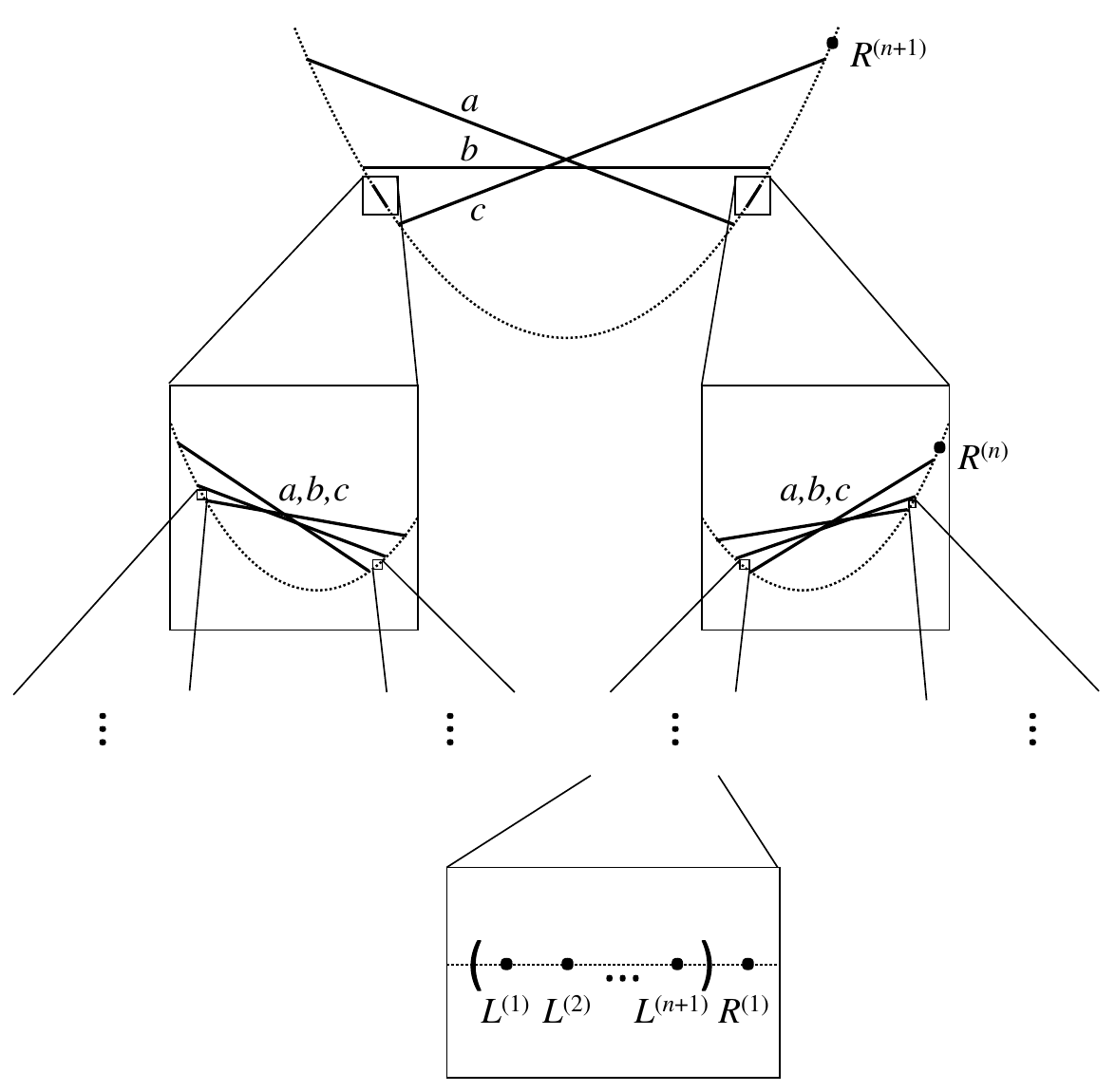}}
\caption{\label{fig_T}The configuration $T_n$ contains \emph{whiskers} (triples of segments labeled $a$, $b$, $c$) organized in a nested, binary-tree structure. In addition, it contains \emph{numeric segments}, grouped into sets of size $n+1$. This sketch shows only one such set of numeric segments, and it only shows their endpoints. For each set of numeric segments, their left endpoints lie in a special block of length $n+1$ (indicated by parentheses in the figure).}
\end{figure}

Consider one such special block $(L^{(1)}, \ldots, L^{(n+1)})$ (where the endpoints have been renamed for clarity of exposition). The corresponding right endpoints are located as follows: Let $w_1, \ldots, w_n$ be the whiskers containing our special block, where $w_1$ is the innermost whisker and $w_n$ is the outermost one. Then, the first right endpoint $R^{(1)}$ appears within $w_1$, right after $L^{(n+1)}$. For each $2\le i\le n$, the right endpoint $R^{(i)}$ appears outside $w_{i-1}$ but inside $w_i$. Finally, $R^{(n+1)}$ appears outside $w_n$. See Figure~\ref{fig_T} for a rough sketch of $T_n$.

\begin{lemma}\label{lemma_one_wide}
Suppose that, in $T_n$, the segments $a,b,c$ in each whisker intersect concavely. Then at least one the sets of $n+1$ numeric segments must be wide.
\end{lemma}

\begin{proof}
Let $w=(a,b,c)$ be the outermost whisker of $T_n$. By the third claim of Lemma~\ref{lem_ratios}, we have $R_{bx} - R_{ax} < R_{cx} - R_{bx}$ or $L_{cx}-L_{bx} < R_{cx} - L_{cx}$ (or both). In the first case, we proceed to the left branch in the binary tree of whiskers. In the second case, we proceed to the right branch. Let $w'=(a',b',c')$ be the next whisker in the chosen branch. We invoke Lemma~\ref{lem_ratios} again on $w'$, and we proceed in this way, going down the binary tree of whiskers until we reach a ``leaf" whisker $w^*$. We invoke Lemma~\ref{lem_ratios} one final time on $w^*$ in the same way, and we choose one of the two special blocks that lie within $w^*$. The $n+1$ numeric segments in that special block form a wide set, as can be easily checked.
\end{proof}

\begin{remark}\label{remark_binary_tree}
As we saw, in the whisker $w=(a,b,c)$ of $Z_m$, one of two gaps is shorter than the distance from the end of the gap to the end of the whisker. If we could find a configuration in which a \emph{specific} gap is shorter than the distance to the end of the configuration, we could obviate the need for a binary tree and significantly reduce the size of our construction.
\end{remark}

Next, we show how a wide set of $5$ segments can lead to trouble:

\begin{lemma}\label{lem_smash_wide}
Consider the sequence
\begin{equation*}
Y\circ F_5 = L_d\, L_e\, L_1\, L_2\, L_f\, R_d\, L_3\, L_4\, R_e\, R_f\, L_5\, R_1\, R_2\, R_3\, R_4\, R_5.
\end{equation*}
It is impossible for the segments $d,e,f$ to intersect concavely, and for the segments $1,2,3,4,5$ to form a wide set and intersect concavely, all at the same time.
\end{lemma}

\begin{proof}
Suppose for a contradiction that we have a realization of such a configuration. Applying Lemma~\ref{lemma_wide_fan} on the segments $1,2,3,4,5$ (with $k=1$ and $k=3$), we have
\begin{equation*}
L_{2x} - L_{1x} > L_{5x} - L_{2x}, \qquad L_{4x} - L_{3x} > L_{5x} - L_{4x}.
\end{equation*}
Hence, we have both
\begin{equation*}
R_{ex}-R_{dx} > L_{4x} - L_{3x} > L_{5x} - L_{4x} > R_{fx} - R_{ex}
\end{equation*}
and
\begin{equation*}
L_{fx} - L_{ex} > L_{2x} - L_{1x} > L_{5x} - L_{2x} > R_{fx} - L_{fx},
\end{equation*}
contradicting the third claim of Lemma~\ref{lem_ratios} on the segments $d,e,f$.
\end{proof}

Now we can put everything together. Define the endpoint sequence
\begin{equation*}
X = Y \circ T_4.
\end{equation*}
$X$ contains $15$ whiskers, $16$ groups of segments of type $d,e,f$, and $16$ groups of numeric segments, with five segments in each group. Hence, $X$ contains a total of $173$ segments.

\begin{theorem}\label{thm_forbidden_2}
It is impossible to realize $X$ such that the segments $a,b,c$ in each whisker intersect concavely,  the segments $d,e,f$ in each copy of $Y$ intersect concavely, and the five numeric segments in each group of numeric segments intersect concavely.
\end{theorem}

\begin{proof}
In $X$, each group of numeric segments of $T_4$ is shuffled into a copy of $Y$ as in the premise of Lemma~\ref{lem_smash_wide}. But by Lemma~\ref{lemma_one_wide}, one of these groups must form a wide set. This is impossible by Lemma~\ref{lem_smash_wide}.
\end{proof}

\section{The Hart--Sharir sequences are unrealizable}\label{sec_ababa}

In this section we show that the Hart--Sharir sequences eventually force a copy of the configuration $X$ whose segments intersect in the way specified in Theorem~\ref{thm_forbidden_2}. (Recall the definition of \emph{forcing} a segment configuration from Section~\ref{subsec_our_results}.) Hence, the Hart--Sharir sequences cannot be realized as lower inner-zone sequences of a parabola.

\subsection{The Hart--Sharir sequences}

We first recall the Hart--Sharir construction~\cite{HS} of superlinear-length order-$3$ DS sequences.

The Hart--Sharir sequences form a two-dimensional array $S_k(m)$, for $k,m\ge 1$; they satisfy the following properties:
\begin{itemize}
\item Some contiguous subsequences (\emph{blocks}) in $S_k(m)$ are designated as \emph{special blocks}.
\item All special blocks in $S_k(m)$ have length exactly $m$.
\item Each symbol in $S_k(m)$ makes it first occurrence in a special block, and each special block contains only first occurrences of symbols.
\item $S_k(m)$ contains no adjacent repetitions and no alternation $ababa$.
\item Each symbol in $S_k(m)$ occurs at least twice (unless $k=m=1$).
\end{itemize}

The construction uses an operation called \emph{shuffling}, which, as we will see, is very closely related to the \emph{endpoint shuffling} defined in Section~\ref{sec_forbidden_2}.

\begin{definition}
Let $A$ be a sequence that has $k$ special blocks of length $m$, and let $B$ be a sequence that has $\ell$ special blocks of length $k$. Then the \emph{shuffle} of $A$ and $B$, denoted $A\bullet B$, is a new sequence having $k\ell$ special blocks of length $m+1$, formed as follows: We make $\ell$ copies of $A$ (one for each special block of $B$), each one having ``fresh" symbols that do not occur in $B$ or in any other copy of $A$.

For each special block $\Gamma_i = (a_1 a_2 \cdots a_k)$ in $B$, $1\le i \le \ell$, let $A_i$ be the $i$th copy of $A$. For each special block $\Delta_j$ in $A_i$, $1\le j\le k$, we insert the symbol $a_j$ at the end of $\Delta_j$ (so its length grows from $m$ to $m+1$) and we duplicate the $m$th symbol of $\Delta_j$ immediately after $\Delta_j$. Then we place another copy of $a_k$ immediately after $A_i$. Call the resulting sequence $A'_i$.

Then $A\bullet B$ is obtained from $B$ by replacing each special block $\Gamma_i$ in it by $A'_i$.
\end{definition}

In the construction of $A\bullet B$, the symbols of the copies of $A$ are usually called \emph{local}, while the symbols of $B$ are called \emph{global}.

For example, let $A = (a)(b)(c)babc$ and $B=(123)21(456)5414525636$. Then,
\begin{equation*}
A\bullet B = (a1)a(b2)b(c3)cbabc3\ 21\ (a'4)a'(b'5)b'(c'6)c'b'a'b'c'6\ 5414525636.
\end{equation*}

Let $S$ be a sequence of symbols in which the first occurrences of the symbols appear in special blocks, and in which each symbol occurs at least twice. Let us extend the definition of the \emph{endpoint sequence} $E(S)$ from Section~\ref{sec_preliminaries}, by specifying that $E(S)$ has special blocks of left endpoints, which are inherited from the special blocks of $S$ in the natural way.

\begin{observation}\label{ref_obs_shuffling}
The relation between shuffling and endpoint shuffling is as follows:
\begin{equation*}
E(A\bullet B) = E(A) \circ E(B).
\end{equation*}
\end{observation}

\begin{proof}
This is immediate from the definitions of $\bullet$ and $\circ$.
\end{proof}

The Hart--Sharir sequences $S_k(m)$ are defined by double induction. The base cases are $k=1$ and $m=1$. For $k=1$ we let
\begin{equation*}
S_1(m) = \Bigl(12\cdots (m-1)m\Bigr)\, (m-1)\cdots 212\cdots m,
\end{equation*}
be an $N$-shaped sequence having a single special block of length $m$. (Actually, of all sequences $S_1(m)$, the construction only uses $S_1(2) = (12)12$.)

For $m=1$, $k\ge 2$, we let $S_k(1)$ be equal to $S_{k-1}(2)$, but with each special block of size $2$ split into two special blocks of size $1$.

Finally, for $m,k\ge 2$, we let
\begin{equation*}
S_k(m) = S_k(m-1) \bullet S_{k-1}(N),
\end{equation*}
where $N$ is the number of special blocks in $S_k(m-1)$.

Thus, we have, up to a renaming of the symbols,
\begin{align*}
S_2(1) &= (1)(2)12,\\
S_2(2) &= (12)1(34)313424,\\
S_2(3) &= (123)21(456)5414525636,\\
S_2(4) &= (1234)321(5678)76515626737848,\\
\vdots\\
S_3(1) &= (1)(2)1(3)(4)313424,\\
S_3(2) &= (12)1(34)31(56)5(78)75157378642\\
&\qquad(9A)9(BC)B9(DE)D(FG)FD9DFBFGECA2AC4CE6EG8G,\\
\vdots\\
S_4(1) &= (1)(2)1(3)(4)31(5)(6)5(7)(8)75157378642\\
&\qquad(9)(A)9(B)(C)B9(D)(E)D(F)(G)FD9DFBFGECA2AC4CE6EG8G,\\
\vdots
\end{align*}
Note that, in the construction of $S_k(m)$, the special blocks of $S_{k-1}(N)$ ``dissolve", and the only special blocks present in $S_k(m)$ are those that come from the copies of $S_k(m-1)$ (enlarged by one).

\subsection{Properties}

Here we establish some important properties of $S_k(m)$.

\begin{lemma}\label{lemma_SB_N}
Each special block $\bigl(1\cdots m\bigr)$ in $S_k(m)$ is immediately followed by $(m-1)\cdots 1$, and followed later on by $\cdots 2\cdots 3\cdots\ \cdots\ m$, thus forming an $N$-shaped subsequence.
\end{lemma}

\begin{proof}
By induction. The claim is true for $k=1$ and all $m$. If the claim is true for $S_{k-1}(2)$, then it is also true for $S_k(1)$. Now, let $k,m\ge 2$, and suppose the claim is true for $S_k(m-1)$. Then, by the definition of $\bullet$, and since each symbol in $S_{k-1}(N)$ occurs at least twice, the claim is also true for $S_k(m) = S_k(m-1) \bullet S_{k-1}(N)$.
\end{proof}

\begin{lemma}
$S_k(m)$ does not contain $ababa$ (Hart, Sharir~\cite{HS}) nor $abcaccbc$ (Pettie~\cite{pettie_origins}).\footnote{Note that $S_3(2)$ already contains $(abcaccbc)^R$.}
\end{lemma}

\begin{proof}[Proof sketch]
For the first claim, suppose for a contradiction that $k$ and $m$ are minimal such that $S_k(m)$ contains $ababa$. Recall that $S_k(m) = S_k(m-1)\bullet S_{k-1}(N)$. Each of the symbols $a$, $b$ must have come either from a copy of $S_k(m-1)$ (in which case it is a local symbol) or from $S_{k-1}(N)$ (in which case it is a global symbol). A case analysis shows that none of the possibilities work. For example, it cannot be that $a$ is local and $b$ is global, because then there would be at most one $b$ between two $a$'s. It cannot be either that $a$ is global and $b$ is local, because then only the \emph{first} $a$ could appear between two $b$'s.

For the second claim, first note that Lemma~\ref{lemma_SB_N} implies that $S_k(m)$ cannot contain $(bc)ccb$ (with the first $b$ and $c$ in the same special block), since, otherwise, $S_k(m)$ would contain $bcbcb$. For the same reason, if $S_k(m)$ contains $(bc)cb$, then $c$ is \emph{not} the last symbol in the special block.

Now, suppose for a contradiction that $k,m$ are minimal such that $S_k(m)$ contains $abcaccbc$. There are eight possibilities as to which of the symbols $a,b,c$ are local and which are global. We rule them all out by a case analysis. First of all, if $a,b,c$ are all local (resp.~all global), then the only possibility would be for $S_k(m-1)$ (resp.~$S_{k-1}(N)$) to contain $abcacbc$, so that the second $c$ gets duplicated into $cc$ in $S_k(m)$. But this cannot happen, since only \emph{first} occurrences of symbols get duplicated. The case of $a$ being local and $b,c$ being global is ruled out by the considerations in the previous paragraph. The remaining cases are also readily ruled out.
\end{proof}

\begin{definition}\label{def_structurally}
Let $A$ and $B$ be two sequences in which the first occurrences of symbols appear in special blocks. Then we say that $B$ \emph{structurally contains} $A$ if $B$ contains a subsequence $A'$ that not only is isomorphic to $A$, but for every two symbols in $A'$, their first occurrences lie in the same special block if and only if the corresponding symbols in $A$ lie in the same special block.
\end{definition}

We now want to prove that $S_{k'}(m')$ structurally contains $S_k(m)$ for all $k'\ge k$, $m'\ge m$. Hence, any problematic configuration that arises in some $S_k(m)$ will also be present in all subsequent $S_{k'}(m')$.

The proof is somewhat delicate. It is clear, for example, that $S_{k+1}(m)$ contains a sequence isomorphic to $S_k(m)$: $S_{k+1}(m)$ contains a global copy of $S_k(N)$ for some $N$ larger than $m$, and in turn $S_k(N)$ contains many local copies of $S_k(N-1)$, et cetera. This simple observation, however, does not imply that $S_{k+1}(m)$ \emph{structurally} contains $S_k(m)$, because the copy of $S_k(m)$ that we found in $S_{k+1}(m)$ has its special blocks completely ``dissolved". To prove that $S_{k+1}(m)$ structurally contains $S_k(m)$ we have to work a bit more carefully.

\begin{definition}\label{def_rank}
The \emph{rank} of a symbol $a$ in $S_k(m)$ is the position (between $1$ and $m$) that the first occurrence of $a$ occupies within its special block.
\end{definition}

Thus, the local symbols of $S_k(m)$ are those with ranks $1,\ldots,m-1$, and the global symbols are those with rank $m$.

\begin{lemma}\label{lem_structurally}
For every $k'\ge k$, $m'\ge m$, and for every choice of $m$ ranks $1\le r_1 < r_2 < \cdots < r_m \le m'$, the sequence $S_{k'}(m')$ structurally contains $S_k(m)$ using symbols of these ranks.
\end{lemma}

\begin{proof}
Denote $B = S_k(m)$ and $D= S_{k'}(m')$, and let $N_2$ and $N_2'$ denote, respectively, the number of special blocks in $B$ and $D$. Note that, in order to specify how $B$ lies within $D$, all we have to do is specify which $N_2$ special blocks $1\le b_1 < b_2 < \cdots < b_{N_2} \le N_2'$ of $D$ we take the symbols from.

We will first take care of the base cases $k = 1$ and $m = 1$.

If $k=1$ then the claim follows by Lemma~\ref{lemma_SB_N}. If $m = m'= 1$ then there is nothing to prove, since in this case structural containment is the same as regular containment. Now suppose $m=1$ and $m'\ge 2$. Recall that in this case we have $B = S_{k-1}(2)$. We are given a rank $r_1$. The symbols of rank $r_1$ in $D$ are the global symbols of $S_{k'}(r_1)$. The global sequence used in forming $S_{k'}(r_1)$ is $S_{k'-1}(N)$ for some $N$. This latter sequence contains $B = S_{k-1}(2)$ in the regular sense, which is enough for us.

Now suppose $k,m\ge 2$.

For convenience let $A = S_k(m-1)$, $C = S_{k'}(m'-1)$. Let $N$ and $N'$ be, respectively, the number of special blocks in $A$ and $C$. Let $X = S_{k-1}(N)$, $Y = S_{k'-1}(N')$. Hence,
\begin{align*}
B &= A \bullet X,\\
D &= C \bullet Y.
\end{align*}

We assume by induction that $C$ contains $A$, and that $Y$ contains $X$, in the stronger sense of our lemma. If $m'>m$ we can also assume by induction that $C$ contains $B$ in this stronger sense. Our objective, as we said, is to show that $D$ contains $B$ in this stronger sense.

We consider two cases: If $r_m < m'$ (which implies $m<m'$), then we wish to find $B$ using only local symbols of $D$. But we know by induction that $C$ (which is structurally contained in $D$) contains $B$ in the desired way.

The second case is if $r_m = m'$. We know by induction that $C$ contains $A$ using ranks $r_1, \ldots, r_{m-1}$. Let $1\le b_1 < b_2 < \cdots < b_{N}\le N'$ be the special blocks of $C$ from which we take the symbols that form $A$ this way.

Next, find $X$ in $Y$ using ranks $b_1, b_2, \ldots, b_N$. We know this is possible by induction. Let $1\le c_1 < c_2 < \cdots$ be the special blocks of $Y$ from which we take the symbols that form $X$ this way.

Now, in order to find $B$ in $D$, take only the local copies numbered $c_1, c_2, \ldots$ of $C$, and within each local copy of $C$, take only the special blocks numbered $b_1, b_2, \ldots, b_{N'}$. Then the symbols of $Y$ that are shuffled into these special blocks are exactly those that form $X$.

Hence, we exactly mimic the construction of $B$ from $A$ and $X$ inside the construction of $D$ from $C$ and $Y$---with only one exception: In the construction of $B$, we duplicate the last symbols of the special blocks of $A$ and $X$. This might not happen to the corresponding symbols in $D$, if these symbols are not the last symbols of the special blocks of $C$ and $Y$. However, we do not need these duplications, since they are \emph{already} present immediately after the special blocks, by Lemma~\ref{lemma_SB_N}.
\end{proof}

\subsection{Clamped symbols in the Hart--Sharir sequences}

Recall the definition of left- and right-clamped symbols (Definition~\ref{def_clamped}) and its use in Lemma~\ref{lemma_clamped}. We now proceed to show that most symbols in $S_k(m)$ are left- and right-clamped.

\begin{lemma}\label{lem_leader}
All symbols in $S_k(m)$ that have rank at least $2$ are left-clamped, and all the symbols, other than the very last symbol in $S_k(m)$, are right-clamped.
\end{lemma}

\begin{proof}
The first claim follows immediately from Lemma~\ref{lemma_SB_N}.

For the second claim, recall that $S_k(m) = S_k(m-1)\bullet S_{k-1}(N)$. If $a$ is the very last symbol of $S_{k-1}(N)$, then it is the very last symbol of $S_k(m)$ and there is nothing to prove. If $a$ is the very last symbol of a copy $S'$ of $S_k(m-1)$, then its last occurrence in $S_k(m)$ is immediately followed by a global symbol $b$, whose first occurrence was shuffled at the end of the \emph{last} special block of $S'$. The first occurrence of $a$ must lie further to the left (also in $S'$).

In any other case, the claim follows by induction on $S_k(m-1)$ or $S_{k-1}(N)$, depending on whether $a$ is a local or a global symbol, since the shuffling operation does not separate the last occurrence of any symbol from the immediately following symbol.
\end{proof}

\begin{corollary}\label{cor_global_clamped}
Once a copy of $S_k(m)$ is shuffled into another sequence ($S_{k+1}(m')$ for some $m'>1$), all its symbols are left-clamped, and all its symbols except for the very last one are right-clamped.
\end{corollary}

\begin{proof}
All the symbols of the copy of $S_k(m)$ that resides in $S_{k+1}(m')$ are global (their rank is $m'$).
\end{proof}

We can also say something about left-clamped rank-$1$ symbols:

\begin{lemma}\label{lem_rank1_leftclamped}
Let $k, m \ge 2$, and let $a$ be a rank-$1$ symbol in $S_k(m) = S_k(m-1)\bullet S_{k-1}(N)$. If $a$ was left-clamped in its local copy of $S_k(m-1)$, then it is also left-clamped in $S_k(m)$.
\end{lemma}

\begin{proof}
By Lemma~\ref{lemma_SB_N}, $S_k(m-1)$ contains adjacent special blocks only if $m=2$. Hence, for $m\ge 3$, the symbol $b$ immediately preceding the first $a$ in $S_k(m-1)$ does not belong to a special block, so $b$ remains adjacent to that $a$ in $S_k(m)$. For $m=2$, the said symbol $b$ might belong to a special block. But then, in the construction of $S_k(m)$, a copy of $b$ is created and placed right before the first $a$.
\end{proof}

\subsection{Geometric unrealizability}

We are now ready to show that the Hart--Sharir sequences force the impossible configuration $X$ of Section~\ref{sec_forbidden_2}.

\begin{theorem}\label{thm_HS_forces_X}
There exists an $m$ such that the Hart--Sharir sequence $S_7(m)$ forces a copy of the configuration $X$ whose segments intersect in the way specified in Theorem~\ref{thm_forbidden_2}. Hence, if $S$ is the lower inner-zone sequence of the parabola $C$, then $S$ cannot contain $S_7(m)$ as a subsequence.\footnote{It should probably be enough to take $m=2$ in the theorem, but the calculations do not seem to be worth the effort.}
\end{theorem}

\begin{proof}

Recall that in Section~\ref{sec_forbidden_2} we defined
\begin{equation*}
F_m = (L_1\cdots L_m)\quad R_1\cdots R_m.
\end{equation*}
Let
\begin{equation*}
Z_{j,m} = L_{a_1}\, F_m\, L_{a_2}\, F_m\, \cdots \, F_m\, L_{a_j}\quad R_{a_1}\, F_m\, R_{a_2}\, F_m\, \cdots\, F_m\, R_{a_j}
\end{equation*}
(where the $2j-2$ copies of $F_m$ use distinct symbols). Note that the $Z_m$ of Section~\ref{sec_forbidden_2} is contained in $Z_{j,m}$ for $j\ge 3$. Following this correspondence, let us call the $j$-tuple of segments $a_1, \ldots, a_j$ in $Z_{j,m}$ a \emph{whisker}, and let us call the segments in the copies of $F_m$ \emph{numeric segments}.

Define
\begin{equation*}
T_{j,n} = ( \cdots (((Z_{j,1} \circ Z_{j,2j-2}) \circ Z_{j,(2j-2)^2}) \circ Z_{j,(2j-2)^3}) \circ \cdots Z_{j,(2j-2)^{n-1}}) \circ F_{(2j-2)^n}.
\end{equation*}

For $j\ge 3$, $T_{j,n}$ is a supersequence of the endpoint sequence $T_n$ that we defined in Section~\ref{sec_forbidden_2}. $T_{j,n}$ contains whiskers that are nested inside one another in the form of a full $(2j-2)$-ary tree of height $n$. In addition, it contains numeric segments, whose left endpoints are grouped into special blocks of length $n+1$, and whose right endpoints are located in a manner analogous to the one described in Section~\ref{sec_forbidden_2}.

Our first goal is to show that the Hart--Sharir sequences force segment configurations of the form $T_{j,n}$ for arbitrarily large $j$ and $n$, in which the $j$ segments in each whisker intersect concavely. From Lemma~\ref{lemma_SB_N} and Observation~\ref{obs_Nshaped}, it follows that the $n+1$ numeric segments in each special block must also intersect concavely.

\begin{observation}\label{obs_S1}
$E(S_1(m)) = F_m$, and $E(S_2(m))$ contains
\begin{equation}\label{eq_obs_S1}
(L_1 \cdots L_m)\,(\,)\,R_1\cdots R_m.
\end{equation}
\end{observation}

\begin{corollary}\label{cor_S3}
$E(S_3(m))$ contains
\begin{equation}\label{eq_cor_S3}
(L_1\,L_2 \cdots L_m)\,(\,)\,R_1\,(\,)\,R_2\,(\,)\cdots(\,)\,R_m.
\end{equation}
\end{corollary}

\begin{proof}
By induction on $m$. Consider $S_3(m) = S_3(m-1)\bullet S_2(N)$, and assume by induction that $S_3(m-1)$ contains an instance $A$ of (\ref{eq_cor_S3}) with $m-1$ in place of $m$. Consider an instance $B$ of (\ref{eq_obs_S1}) in $E(S_2(N))$ with $N$ in place of $m$. The first special block of $B$ is shuffled into a copy of $S_3(m-1)$; hence, one of its symbols $L^*$ is inserted at the end of the first special block of $A$ in this copy of $S_3(m-1)$. Furthermore, the $(\,)$ and the corresponding symbol $R^*$ of $B$ are placed after this copy of $A$. Hence, $S_3(m)$ contains (\ref{eq_cor_S3}). 
\end{proof}

\begin{corollary}\label{cor_S4}
For each $m\ge 2$, $S_4(m)$ forces a configuration of the form
\begin{equation}\label{eq_cor_S4}
L_1\,(\,)\,L_2\,(\,)\cdots(\,)\,L_N\,(\,)\,R_1\,(\,)\,R_2\,(\,)\cdots(\,)\,R_N
\end{equation}
for some very large $N=N(m)$, in which the segments $1,2,\ldots, N$ intersect concavely and have rank $m$.
\end{corollary}

\begin{proof}
The global sequence $S_3(N')$ used in forming $S_4(m)$ satisfies Corollary~\ref{cor_S3}. The left endpoints $L_1, L_2, \ldots$ of (\ref{eq_cor_S3}) receive rank $m$ and go into separate special blocks. In order to guarantee the presence of a special block between $L_i$ and $L_{i+1}$ for each $i$, we ``sacrifice" every second symbol among $1,2,\ldots,N'$; we are still left with $N=N'/2$ symbols. The $N$-shaped sequence $1\cdots N'\cdots 1\cdots N'$ that was present in $S_3(N')$ (by Lemma~\ref{lemma_SB_N}) is obviously still present in $S_4(m)$.
\end{proof}

\begin{corollary}\label{cor_S5}
For each $m\ge 2$, $S_5(m)$ forces a configuration of the form $Z_{N,m-1}$ for some very large $N = N(m)$, in which the whisker segments $a_1, \ldots,\allowbreak a_N$ intersect concavely, and in which the numeric segments of each copy of $F_{m-1}$ have ranks $1,\ldots,m-1$.
\end{corollary}

\begin{proof}
The global sequence $S_4(N')$ used in forming $S_5(m)$ satisfies Corollary~\ref{cor_S4}. Each special block of (\ref{eq_cor_S4}) is replaced by a copy of $S_5(m-1)$, which structurally contains $S_1(m-1)$.
\end{proof}

Finally:

\begin{corollary}\label{cor_S6}
For every $j$ and $n$, if $m$ is large enough, then $S_6(m)$ forces a configuration of the form $T_{j,n}$, in which the segments $a_1, \ldots, a_j$ in each whisker intersect concavely.
\end{corollary}

\begin{proof}
The iterated shuffling used to form $T_{j,n}$ occurs naturally in the formation of $S_6(m) = S_6(m-1) \bullet S_5(N)$.

More precisely, we will show by induction on $n$ that, if $m$ is large enough in terms of $j$ and $n$, then $S_6(m)$ forces a configuration of the form
\begin{equation*}
T'_{j,n} = ( \cdots ((Z_{j,1} \circ Z_{j,2j-2}) \circ Z_{j,(2j-2)^2}) \circ \cdots )\circ Z_{j,(2j-2)^n},
\end{equation*}
in which the segments of each whisker intersect concavely. Since $Z_{j,(2j-2)^n}$ contains $F_{(2j-2)^n}$, it will follow that $T'_{j,n}$ contains $T_{j,n}$.

Suppose by induction that, for some $m$, the sequence $S_6(m)$ forces such a configuration $T'_{j,n-1}$. Let $N$ be the number of special blocks of $S_6(m)$, and let $1\le b_1 < b_2 < \cdots < b_k\le N$, for $k=(2j-2)^n$, be the special blocks of $S_6(m)$ involved in this occurrence of $T'_{j,n-1}$. By Corollary~\ref{cor_S5}, there exists an $N'$ such that $S_5(N')$ forces a copy of $Z_{j,b_k}$ in which the numeric segments of each $F_{b_k}$ have ranks $1,\ldots, b_k$. Out of the $b_k$ segments in each $F_{b_k}$, we are only interested in the ones ranked $b_1, \ldots, b_k$: They are the ones that will be shuffled into the right places in $T'_{j,n-1}$. Let $m'\ge m$ and $N''\ge N'$ be such that $S_6(m'+1) = S_6(m')\bullet S_5(N'')$. Now, $S_6(m')$ contains a copy of $S_6(m)$ using the first $N$ special blocks of $S_6(m')$. And the above-mentioned copy of $Z_{j,b_k}$ that is present in $S_5(N')$ is also present in $S_5(N'')$. Hence, $S_6(m'+1)$ forces the desired copy of $T'_{j,n}$.
\end{proof}

Hence, as claimed, the sequences $S_6(m)$ force arbitrarily wide and tall configurations of the form $T_{j,n}$, in which the segments of each whisker intersect concavely. In particular, they force one such copy of $T_{3,4}$, which contains the configuration $T_4$ of Section~\ref{sec_forbidden_2}.

Our second goal is to show that this sequence $T_{3,4}$ is appropriately shuffled into a sequence containing $Y$ (which was defined in Section~\ref{sec_forbidden_2}).

We observe that (\ref{eq_cor_S4}) already contains $Y$ whenever $N\ge 6$. Hence, $S_7(1)$ also contains $Y$. Let $b_1 < b_2 < \cdots < b_5$ be the special blocks of $S_7(1)$ involved in this occurrence of $Y$. By Corollary~\ref{cor_S6} and Lemma~\ref{lem_structurally}, $S_6(m)$ contains, for large enough $m$, a copy of $T_{3,4}$ in which the numeric segments have ranks $b_1, b_2, \ldots, b_5$. Therefore, there exists an $m$ (probably $m=2$ should be enough) such that, in $S_7(m) = S_7(m-1) \bullet S_6(N)$, these numeric segments are shuffled into the right places, creating the desired copy of $X$. In this copy of $X$, all the sets of segments that should intersect concavely according to Theorem~\ref{thm_forbidden_2}, actually do.

Finally, all the symbols in this copy of $X$ are left- and right-clamped in $S_7(m)$: The copy of $Y$ in $S_7(1) = S_6(2)$ uses symbols that, in $S_6(2)$, had rank $2$, so they were clamped by Lemma~\ref{lem_leader}. Hence, by Lemma~\ref{lem_rank1_leftclamped}, they stay clamped in $S_7(m)$. And the symbols of $S_6(N)$ become global  in $S_7(2)$, so they are also clamped by Lemma~\ref{lem_leader}. Hence, by Lemma~\ref{lemma_clamped}, if $S$ is a lower inner-zone sequence that contains $S_7(m)$, then the endpoints of $X$ appear in the right order in $E(S)$.
\end{proof}

\subsection{How we found these results}

Our results of Sections~\ref{sec_forbidden_2} and~\ref{sec_ababa} were obtained as follows:

Matou\v sek~\cite{mat_DG} and Sharir and Agarwal~\cite{DS_book} describe a construction by P. Shor of segments in the plane whose lower-envelope sequences are the Hart--Sharir sequences. We tried to force the segment endpoints in the construction to lie on the parabola $C$. We managed to do this for $S_1(m)$, $S_2(m)$, and $S_3(m)$ and all $m$, but for $S_4(m)$ this seems impossible.

Shor's construction is based on \emph{fans}---sets of segments that intersect concavely, whose left endpoints are very close to one another, and whose lengths increase very rapidly.  Global fans have their left endpoints shuffled into tiny local fans. In order for the global fan to intersect concavely, its segments are given slopes $1, 1+\eps_1, 1+\eps_2, \ldots$ for very small values of $\eps_1, \eps_2, \ldots$. This gives a lot of freedom to play with the exact position of the segments' left endpoints.

However, if we want all endpoints to lie on a parabola, then the slopes in the global fan must increase very rapidly, which leads to the absurd requirement that the distances between the left endpoints decrease very rapidly (Lemma~\ref{lemma_wide_fan} above). Then it is impossible to appropriately shuffle the global fan into the local fans.

The impossible segment configuration $X$ of Section~\ref{sec_forbidden_2} (which, as we saw in Section~\ref{sec_ababa}, is forced by $S_7(m)$) is the best way we found to isolate the contradiction.

It would be nice to be able to isolate a smaller impossible segment configuration forced by the Hart-Sharir sequences (say, by $S_4(m)$). However, it is unlikely that such an improvement would be of additional help in proving Conjecture~\ref{conj}. The most promising line of attack is described in the next section.

\section{Directions for future work}\label{sec_conj_only_DS}

We believe that our geometric results are sufficient to prove Conjecture~\ref{conj}, and that the remaining work is purely combinatorial:

\begin{conjecture}\label{conj_only_HS}
The Hart--Sharir sequences are the only way to achieve superlinear-length $ababa$-free sequences. Namely, for every Hart--Sharir sequence $S_k(m)$ we have 
\begin{equation*}
\Ex\bigl(\bigl\{ababa, S_k(m), (S_k(m))^R\bigr\}, n\bigr)=O(n);
\end{equation*}
where the hidden constant depends on $k$ and $m$.
\end{conjecture}

In order to establish Conjecture~\ref{conj}, it is enough to prove Conjecture~\ref{conj_only_HS} for the specific (gigantic) case of Theorem~\ref{thm_HS_forces_X}. However, our hope is that Conjecture~\ref{conj_only_HS} can be somehow more easily proven for \emph{all} $k$ and $m$ by a double induction argument.

Conjecture~\ref{conj_only_HS} is known to be true for $k=1$, since $S_1(m)$ are $N$-shaped sequences: Klazar and Valtr~\cite{KV} showed that (even without forbidding $ababa$) we have $\Ex\bigl(S_1(m),n\bigr) \le c_m n$ for some constants $c_m$. Pettie~\cite{pettie_forbid} subsequently improved the dependence of $c_m$ on $m$ to $c_m \le 2^{\Theta(m^2)}$, which is still quite large. No interesting lower bounds for $c_m$ are known. In any case, we conjecture that, forbidding both $ababa$ and an $N$-shaped pattern, we should have $\Ex\bigl(\bigl\{ababa, S_1(m)\bigr\},n\bigr) \le c'_m n$ for some quite small $c'_m$.

\begin{figure}
\centerline{\includegraphics{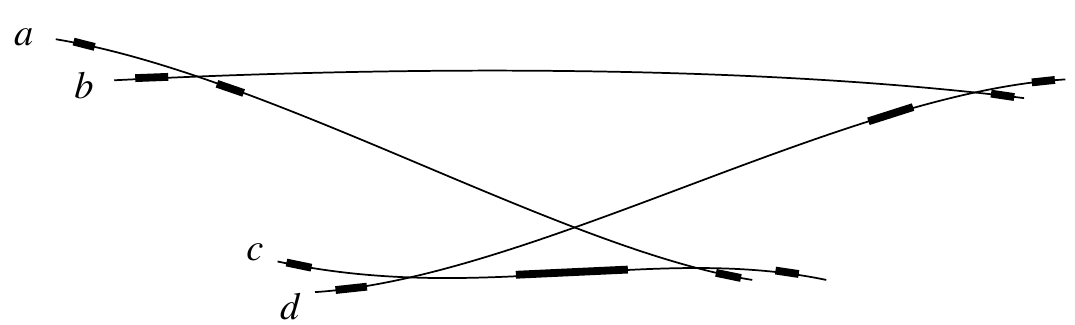}}
\caption{\label{fig_abacdcacdbd} The case $k=m=2$ of Conjecture~\ref{conj_only_HS} states that this pattern must be present in any configuration of $n$ $x$-monotone pseudosegments, if its lower-envelope sequence has length $cn$ for some large enough constant $c$. The highlighted subsegments must be visible from $-\infty$.}
\end{figure}

The first open case in Conjecture~\ref{conj_only_HS} is $k=m=2$. In this case, $S_2(2) = aba\,cdcac\,dbd \equiv (S_2(2))^R$; see Figure~\ref{fig_abacdcacdbd}. However, as we mentioned in the Introduction, even the weaker conjecture, that $\Ex(\{ababa, ab\,cacbc\}, n)=O(n)$, is still open.

(Similarly, the other conjecture mentioned in the Introduction, namely that $\Ex(\{ababa,\allowbreak ab\,cac\,cbc,\allowbreak (ab\,cac\,cbc)^R\}, n)=O(n)$, would follow from the case $k=3$, $m=2$ of Conjecture~\ref{conj_only_HS}.)

\subsection{Linear vs.~nonlinear forbidden patterns}

Conjecture~\ref{conj_only_HS} fits into the following more general question: For which patterns $u$ is $\Ex(u,n)$ linear in $n$? This question has been previously explored in several papers~\cite{AKV,klazar_93,klazar_survey,pettie_forbid,pettie_origins}. The known results in this area are somewhat patchy, and a proof (or disproof) of our conjecture will shed additional light in this area. For one, our conjecture already highlights the fact that forbidding a \emph{set} of patterns might have a stronger effect than forbidding each pattern separately, and hence, the right question should be: For which \emph{sets} of patterns $U$ is $\Ex(U, n)$ linear in $n$?

\section{Conclusion}\label{sec_discussion}

We end by listing some related open problems:

\begin{itemize}
\item As mentioned at the end of Section~\ref{sec_ababa}, it would be nice to isolate a smaller impossible segment configuration forced by the Hart--Sharir sequences (say, by $S_4(m)$ for some $m$); see Remark~\ref{remark_binary_tree} above. (Still, it is unlikely that such an improvement would make proving Conjecture~\ref{conj} much easier.)

\item What if we do not require $C$ to be a circle, but only a convex curve? It still seems impossible to implement Shor's construction forcing the endpoints to lie on a convex curve.

\item  Can the $ababa$-free sequences of Nivasch--Geneson~\cite{geneson,yo_DS} be realized as lower-envelope sequences of line segments? Is there an $ababa$-free sequence that \emph{cannot} be realized in such a way?

\item The longest Davenport--Schinzel sequences of order $4$ ($ababab$-free) have length $\Theta\bigl(n\cdot 2^{\alpha(n)}\bigr)$. However, no one knows how to realize them as lower-envelope sequences of parabolic segments. Perhaps it is impossible. One could start by finding forbidden patterns here.

\item \emph{Higher dimensions:} Raz~\cite{raz} recently proved that the combinatorial complexity of the outer zone of the boundary of a convex body in an arrangement of hyperplanes in $R^d$ is $O(n^{d-1})$. The complexity of the inner zone is only known to be $O(n^{d-1}\log n)$ (Aronov et al.~\cite{APS}). Whether the latter is also linear in $n$ is an open question.
\end{itemize}

\paragraph{Acknowledgements} I would like to give special thanks to Seth Pettie for useful discussions on generalized DS sequences with different forbidden patterns. Thanks also to the referees of previous versions of this work for their careful reading and for their useful suggestions (including the idea in Section~\ref{subsec_conic}). Finally, thanks to Micha Sharir for useful discussions, and to Dan Halperin for encouraging me to work on this problem (several years ago).

\end{document}